%% file: arxiv.tex
\documentclass[12pt]{article}
\input{paper-style}
\input{notation}

\hypersetup{pdftitle={Coarsening Latent-Class Probabilities: Directional Distortion and Coverage Loss}}
\hypersetup{pdfauthor={Marcell T. Kurbucz}}

\newcommand{\citepBinary}{\citealp{kurbucz2026binary}}
\newcommand{\citepKbs}{\citealp{kurbucz2026kbs}}
\newcommand{\citeparenKbs}{\citep{kurbucz2026kbs}}
\newcommand{\citetKbs}{\citet{kurbucz2026kbs}}

\newcommand{\MainLemResid}{1\xspace}
\newcommand{\MainThmIdent}{1\xspace}
\newcommand{\MainThmSpectral}{2\xspace}
\newcommand{\MainPropOrtho}{3\xspace}
\newcommand{\MainThmClt}{4\xspace}
\newcommand{\MainThmEff}{5\xspace}

\newcommand{\MainThmAtten}{7\xspace}
\newcommand{\MainCorAniso}{9\xspace}
\newcommand{\MainPropCoverage}{10\xspace}
\newcommand{\MainPropCal}{11\xspace}

\newcommand{\MainSecSim}{7\xspace}

\newcommand{\MainSecRecal}{9\xspace}
\newcommand{\MainSecBisg}{10\xspace}
\newcommand{\MainEqEst}{8\xspace}
\newcommand{\MainEqVeff}{10\xspace}

\newcommand{\MainEqCovHat}{14\xspace}

\newcommand{\MainTabAdult}{2\xspace}

\begin{document}
\spacingset{1.15}

\input{main-body}

\bibliographystyle{apalike}
\bibliography{references}

\clearpage
\setcounter{section}{0}
\setcounter{theorem}{0}
\setcounter{lemma}{0}
\setcounter{equation}{0}
\setcounter{table}{0}
\setcounter{figure}{0}
\renewcommand{\thetheorem}{S\arabic{theorem}}
\renewcommand{\thesection}{S\arabic{section}}
\renewcommand{\theequation}{S.\arabic{equation}}
\renewcommand{\thetable}{S\arabic{table}}
\renewcommand{\thefigure}{S\arabic{figure}}

\begin{center}
  {\Large\bfseries Supplement to ``Coarsening Latent-Class Probabilities:\\[0.6ex]
  Directional Distortion and Coverage Loss''}\\[1.5ex]
  Marcell T.\ Kurbucz
\end{center}
\bigskip

\input{supplement-body}

\end{document}

%% file: paper-style.tex
\usepackage{amsmath,amssymb,amsthm,mathtools}
\usepackage{graphicx}
\usepackage{enumerate}
\usepackage{float}
\usepackage{booktabs}
\usepackage{threeparttable}
\usepackage{natbib}
\usepackage{xurl}
\usepackage{xspace}
\usepackage[colorlinks=true,linkcolor=blue,citecolor=blue,urlcolor=blue]{hyperref}
\graphicspath{{figures/}}

\newtheorem{theorem}{Theorem}
\newtheorem{proposition}[theorem]{Proposition}
\newtheorem{lemma}{Lemma}
\newtheorem{corollary}[theorem]{Corollary}
\theoremstyle{definition}
\newtheorem{assumption}{Assumption}
\newtheorem{definition}[theorem]{Definition}
\theoremstyle{remark}
\newtheorem{remark}{Remark}

\def\spacingset#1{\renewcommand{\baselinestretch}{#1}\small\normalsize}

%% file: notation.tex
\newcommand{\E}{\mathbb{E}}
\newcommand{\Pbb}{\mathbb{P}}
\newcommand{\Var}{\operatorname{Var}}
\newcommand{\Cov}{\operatorname{Cov}}
\newcommand{\R}{\mathbb{R}}
\newcommand{\N}{\mathcal{N}}
\newcommand{\T}{^{\!\top}}
\newcommand{\pto}{\xrightarrow{\;p\;}}
\newcommand{\dto}{\xrightarrow{\;d\;}}
\newcommand{\plim}{\operatorname{plim}}
\newcommand{\rank}{\operatorname{rank}}
\newcommand{\tr}{\operatorname{tr}}
\newcommand{\range}{\mathcal{R}}
\newcommand{\nullsp}{\mathcal{N}}
\newcommand{\Vs}{V^{*}}
\newcommand{\Msf}{M}
\newcommand{\sgn}{\operatorname{sgn}}
\newcommand{\1}{\mathbf{1}}

%% file: main-body.tex
\begin{center}
  {\LARGE\bfseries Coarsening Latent-Class Probabilities:\\[0.6ex] Directional Distortion and Coverage Loss}\\[2ex]
  Marcell T.\ Kurbucz\\[0.5ex]
  Institute for Global Prosperity, The Bartlett,\\
  University College London\\
  \texttt{m.kurbucz@ucl.ac.uk}\\[0.5ex]
  \today
\end{center}

\begin{abstract}
\noindent Outcomes are regressed on a calibrated probability vector for unobserved class membership. Under a structural conditional mean excluding the score and conditional calibration, the observed-data model reduces to a partially linear regression. The probability vector is a Berkson-type surrogate for membership, so the effect vector $\tau$ is identified without attenuation. In practice the vector is often coarsened to a hard label---an argmax, a confidence threshold---and that label need not retain the Berkson property. For any coarsening the plug-in estimator converges to $\mathcal{A}\tau$, where $\mathcal{A}-I$ is determined by the regression of the discarded part of the score on the retained part. Coarsening therefore leaves $\tau$ undistorted exactly when that regression vanishes, and otherwise distorts some contrasts far more than others. The same operator determines the bias that drives coverage loss. Where that bias is of the order of the standard error, the Wald interval has limiting coverage $\Phi(z-\nu)-\Phi(-z-\nu)$, with $\nu$ their ratio. A fixed bias sends coverage to zero. Operator, standard error, and---through the uncoarsened estimator---the bias are estimable from observed data, so the implied coverage can be approximated before the interval is reported. Simulations show severe coverage loss after argmax coarsening. Three real-data audits exhibit the direction-specific distortion.
\end{abstract}

\noindent\textbf{MSC2020 subject classifications:} Primary 62G20, 62J20; secondary 62P25.


\bigskip

\section{Introduction}
\label{sec:intro}

A recurring problem in empirical work is to measure how membership in a latent group affects an outcome when the group indicator is never observed but a \emph{calibrated probability} for it is. Such probabilities come from a trained classifier or a model-based score: an unobserved protected attribute proxied for a fairness audit \citep{kallus2022,chen2019}, disease status from an imaging classifier, poverty status from administrative predictors. The binary version of this problem admits a closed-form treatment. A scalar calibrated score $p$ satisfying the conditional calibration condition $\E[G\mid p,X]=p$ is a Berkson-type surrogate for the latent indicator $G\in\{0,1\}$. Under a structural mean that excludes the score, the effect $\tau$ is point-identified by a moment equation whose denominator is a residual-score variance. Identification fails exactly when that variance vanishes. This is the two-class ($K=1$) special case of the results developed below (see \citepBinary, for the binary theory, and \citepKbs, for its confidence-thresholding diagnostics).

Most applications, however, are not binary. Fairness audits compare several protected groups; prevalence studies distinguish multiple disease subtypes; classifiers emit a probability vector over $K+1$ mutually exclusive classes. The natural object is then a vector $\tau\in\R^{K}$ of group effects and a vector-valued calibrated score $p$ on the probability simplex. This paper asks: what does coarsening that vector do to the coefficient and to the interval reported from it---along which directions---and can the answer be known before either is reported?

Existing approaches aim at correction rather than at assessing what coarsening does to the reported interval. One line of work documents the damage to an aggregate quantity, where a hard label demonstrably distorts a count or a group mean \citep{dong2025,chen2019}; another corrects the downstream estimate inside a latent-class model, using an estimated matrix of classification errors \citep{bolck2004}; a third restores validity with a gold-standard subsample \citep{angelopoulos2023}. The corrective approaches require something the analyst may not have: a fitted latent-class model, or a gold-standard subsample. None of these approaches directly assesses whether a Wald interval based on the coarsened regression attains its nominal coverage.

The argument rests on a constant-coefficient structural mean and on conditional calibration. Together they imply $\E[Y\mid p,X]=\mu(X)+\tau\T p$, so on observables the problem is a partially linear regression of the outcome on $p$ \citep{robinson1988}. This is the observed-data content of the assumptions: the two restrictions turn unobserved membership into an observable regression, without instruments, repeated measurements or a validation sample.

Define the score residual $a=p-r(X)$ with $r(X)=\E[p\mid X]$, and the \emph{residual-score covariance matrix}
\begin{equation}
  \Msf \;=\; \E\!\left[(p-r(X))(p-r(X))\T\right] \;=\; \E[\Cov(p\mid X)]
  \;\in\;\R^{K\times K}.
  \label{eq:M-intro}
\end{equation}
When $K=1$ this recovers the scalar residual-score variance $\Vs=\E[(p-r(X))^2]$ of the binary case.

In routine practice, however, that regression is not the one an analyst runs. The probability vector is coarsened before use, replaced by an argmax label \citep{lee2013}, by a confidence-thresholded label \citep{sohn2020}, or by a rounded summary. Our first result characterizes the effect of coarsening through an operator. For any such coarsening $h$, the plug-in estimator converges to $\mathcal{A}_h\tau$ with
\begin{equation}
  \mathcal{A}_h \;=\; I + D_h^{-1}\E[a_h u\T],
  \qquad D_h=\E[a_h a_h\T],\quad
  a_h=h-\E[h\mid X],\quad u=a-a_h.
  \label{eq:A-intro}
\end{equation}
The distortion is driven by the covariance between the retained signal $a_h$ and the discarded signal $u$. Coarsening produces no distortion exactly when the two are uncorrelated, and otherwise acts unevenly across directions, so the distortion need not be proportional across contrasts. In the binary case the same object is a scalar dilution factor \citeparenKbs.

The second result translates this distortion into its consequences for the reported interval. An analyst who coarsens and then forms the usual sandwich reports a Wald interval. Along sequences in which the coarsening bias shrinks at the rate the standard error does, its coverage tends to $\Phi(z-\nu_v)-\Phi(-z-\nu_v)$, where $\nu_v$ is the ratio of that bias along $v$ to the reported standard error (Proposition~\ref{prop:coverage}). A bias fixed in $n$ drives coverage to zero. Holding the bias fixed, coverage falls as the reported standard error shrinks, since $\nu_v$ then grows. The quantities needed to approximate this coverage are estimable from the observed data: $\mathcal{A}_h$ is identified from $(p,X)$ alone, the reported standard error from $(Y,p,X)$, and the bias additionally requires $\tau$, which the uncoarsened estimator supplies. That approximation is therefore available in advance, with no class labels. Reading $\mathcal{A}_h\tau$ as a distortion of the latent effect does require the identifying model. In the leading design an argmax interval is $39\%$ narrower than the correct one and covers $1\%$ of the time.

The remaining results tie $\mathcal{A}_h$ to $\tau$ itself rather than to the coefficient from the coarsened regression. Along an eigen-direction with vanishing eigenvalue, the corresponding component of $\tau$ is set-identified rather than point-identified, and the failure is observational rather than an artifact of one moment (Theorem~\ref{thm:spectral}). When calibration itself fails within a budget $\delta$, the induced bias along $v_j$ is bounded by $\delta\|\tau\|_1\lambda_j^{-1/2}$ (Proposition~\ref{prop:cal}), following the approximate-moment-condition approach of \citet{armstrong2021}. The same spectrum that governs precision also bounds sensitivity to miscalibration; coarsening is direction-specific as well, though not ordered by those eigenvalues.

In the proxy-based disparity literature most closely related to our setting, \citet{chen2019} contrast hard-thresholded imputation with a probability-weighted estimator of a binary demographic disparity. \citet{kallus2022} show that when the outcome and the protected class are never observed together, the coupling between them within proxy cells is unrestricted and common disparity measures are only partially identified. Our exclusion restriction constrains that coupling; together with conditional calibration and the rank condition it delivers point identification, and along a collapsed eigen-direction we return to a partially identified regime. \citet{mccartan2025} impose the analogue of our exclusion restriction one step upstream, on the surname rather than on the score, and use it to repair the imputation itself, treating surname as an instrument for membership within a Bayesian model for the disparity. They evaluate the thresholding estimator empirically but do not characterize its bias, noting, with reference to \citet{chen2019}, that its direction is not generally predictable. For any coarsening the distortion operator is the matrix \eqref{eq:A-intro}, identified from $(p,X)$ alone, so its size and directional structure are known without ground-truth class labels.

On North Carolina voter records, \citet{dong2025} attribute $11.9$ of a $28.2$-point undercount of African-American voters in a commercial voter file to argmax labeling rather than to miscalibration, and \citet{xin2026} obtain, to first order, coefficients that mix the true group effects through a confusion matrix. These studies consider different estimands---group-mean disparities, aggregate counts, and, in the last case, a regression coefficient as here---but $\mathcal{A}_h$ is exact, applies to arbitrary coarsenings, and is identified without ground-truth labels. None of this work asks whether the reported interval remains valid.

The misclassification literature identifies regression effects of a mismeasured binary regressor \citep{mahajan2006} and shows that misclassification attenuates them \citep{lewbel2007}; \citet{molinari2008} bounds the distribution of the underlying variable through the matrix of misclassification probabilities. Related results appear in the latent-class literature, where the classify-analyze strategy assigns units to classes from their posterior class probabilities---most simply by taking the class with the largest one. The resulting assignment is then used downstream as though it were the truth, typically yielding attenuated associations. The corrections either invert the classification-error matrix \citep{bolck2004,vermunt2010} or enlarge the classification model \citep{bray2015}. These approaches presuppose either a misclassified label or a fitted latent-class model. We observe a calibrated probability vector, so the coarsening is chosen by the analyst rather than imposed by the data. Consequently, $\mathcal{A}_h$ is estimable without a first-stage model, and the analysis below can assess the reported interval as well as the point estimate.

Recovering the latent structure itself calls for devices unavailable here: diagonalizing integral operators under injectivity conditions \citep{hu2008,schennach2016}, or many conditionally independent measurements \citep{allman2009,kasahara2009}. The data here carry a single probability vector, and the analysis rests instead on the partially linear model \citep{robinson1988} with double machine learning \citep{chernozhukov2018}. A parallel literature corrects inference when the outcome is a machine-learning prediction rather than a measurement \citep{wang2020,miao2025}; here the generated quantity enters as a regressor rather than as an outcome. Recent work on machine-learning-generated regressors \citep{battaglia2024} and on inference using a gold-labeled subsample \citep{angelopoulos2023} addresses a related coverage problem. In both, however, the distortion originates in prediction or estimation error, whereas a deterministic coarsening of an exactly observed vector calls for a diagnostic rather than a correction.

Section~\ref{sec:model} sets up the model and records the reduction to a partially linear regression. Section~\ref{sec:ident} states what that regression identifies and where it fails; Section~\ref{sec:inference} gives orthogonality, the sandwich limit, and the efficiency bound. Section~\ref{sec:atten} introduces the coarsening operator and Section~\ref{sec:coverage} the coverage of the resulting Wald interval. Section~\ref{sec:sim} reports simulations, Section~\ref{sec:adult} a UCI Adult disparity audit, and Section~\ref{sec:recal} a sensitivity bound for miscalibration together with a labeled-subsample experiment probing calibration failure and exclusion violations. Section~\ref{sec:bisg} presents a surname-based voter-turnout audit ($n=464{,}700$). Section~\ref{sec:disc} discusses scope. Proofs, additional experiments, and a land-cover audit are in the supplement.

\section{Model and Assumptions}
\label{sec:model}

We observe i.i.d.\ draws $(Y_i,X_i,p_i)\in\R\times\mathcal{X}\times\Delta_K$, where $\Delta_K=\{p\in[0,1]^{K}:\mathbf 1\T p\le 1\}$ collects the first $K$ coordinates of a $(K{+}1)$-class probability vector; the omitted reference probability is $p_{K+1}=1-\mathbf 1\T p$. On the same space there is an unobserved one-hot membership vector $g\in\{0,1\}^{K}$, $g_k=\mathbf{1}\{G=k\}$, $G\in\{1,\dots,K+1\}$, with $K+1$ the reference class. Write $m(x)=\E[Y\mid X=x]$, $r(x)=\E[p\mid X=x]\in\R^{K}$, and define
\begin{equation}
  a \;=\; p-r(X)\in\R^{K}, \qquad R \;=\; Y-m(X),
  \label{eq:resid}
\end{equation}
so $\E[a\mid X]=0$ and $\E[R\mid X]=0$. The residual-score covariance matrix $\Msf=\E[a a\T]$ of \eqref{eq:M-intro} is symmetric positive semidefinite.

\begin{assumption}[Structural conditional mean]
\label{ass:mean}
There exist a measurable $\mu:\mathcal{X}\to\R$ and a vector $\tau\in\R^{K}$ with $\E[Y\mid G,p,X]=\mu(X)+\tau\T g$ almost surely; here $\tau_k$ is the conditional mean contrast between class $k$ and the reference class $K{+}1$, holding $X$ fixed.
\end{assumption}

\begin{assumption}[Conditional calibration]
\label{ass:cal}
$\E[g\mid p,X]=p$ almost surely.
\end{assumption}

\begin{assumption}[Non-degeneracy]
\label{ass:rank}
$\Msf=\E[(p-r(X))(p-r(X))\T]$ is nonsingular ($\lambda_{\min}(\Msf)>0$).
\end{assumption}

\begin{assumption}[Moments]
\label{ass:mom}
$\E[Y^4]<\infty$ (since $p\in\Delta_K$ is bounded, all moments of $p$ are finite).
\end{assumption}

Assumption~\ref{ass:mean} has two parts: the latent-class effects are constant in $X$, and the probability vector is mean-independent of $Y$ given $(G,X)$. The latter is an \emph{exclusion restriction}. The score may predict membership, but conditional on $G$ and the downstream controls $X$ it carries no further mean-relevant information about $Y$. We use ``effect'' for this structural conditional-mean contrast, not a causal effect absent further assumptions; the supplement treats the marginal contrast, which is also point-identified and differs from $\tau_k$ by an explicit compositional term, and the case of effects varying with $X$. Assumption~\ref{ass:cal} is the sole link between the unobserved $g$ and the observed $p$. It makes $p$ a Berkson-type surrogate for $g$: true membership scatters around the observed score rather than the score around the truth, so regressing $Y$ on $(p,X)$ carries no attenuation and $\tau$ is recovered without bias. This reduction is standard, and the paper takes it as given; the property does not survive coarsening, since $\E[g\mid h(p),X]\neq h(p)$ in general. The condition is stronger than ordinary multiclass calibration: it requires the probability vector to be unbiased for membership after conditioning on $X$, the analogue of an identifying exogeneity condition. Subpopulation miscalibration in $X$ therefore violates it, and Proposition~\ref{prop:cal} shows the resulting bias is amplified by small eigenvalues of $\Msf$. Assumption~\ref{ass:rank} is the multivariate non-degeneracy condition.

\begin{lemma}[Residual decomposition]
\label{lem:resid}
Under Assumptions~\ref{ass:mean} and~\ref{ass:cal},
\begin{equation}
  R \;=\; \tau\T\!\left(g-r(X)\right) + \varepsilon,
  \qquad \E[\varepsilon\mid G,p,X]=0.
  \label{eq:lem-resid}
\end{equation}
\end{lemma}

Lemma~\ref{lem:resid} (proved in the supplement) gives the residual decomposition used throughout: the systematic part of the outcome residual is driven by the deviation $g-r(X)$ of true membership from its conditional expectation given $X$.

\subsection*{Observed-data implications of the assumptions}

Assumptions~\ref{ass:mean} and~\ref{ass:cal} act on the observed data as follows. Taking $\E[\,\cdot\mid p,X\,]$ in Assumption~\ref{ass:mean} and substituting Assumption~\ref{ass:cal},
\begin{equation}
  \E[Y\mid p,X] \;=\; \mu(X)+\tau\T\E[g\mid p,X] \;=\; \mu(X)+\tau\T p.
  \label{eq:plr}
\end{equation}
The latent class has disappeared: on observables the model is a partially linear regression of $Y$ on the probability vector $p$ with nonparametric component $\mu$ \citep{robinson1988}. Everything that follows about point identification, orthogonality and the sandwich limit is the corresponding statement for that regression, with $\Msf$ the covariance matrix of the partialled-out regressor and $\tau$ the partial regression coefficient.

We record \eqref{eq:plr} explicitly because it is what Assumptions~\ref{ass:mean} and~\ref{ass:cal} imply, not a further restriction imposed on top of them. It also defines the scope of the paper. Sections~\ref{sec:ident}--\ref{sec:inference} develop the implications of \eqref{eq:plr} needed later. Section~\ref{sec:atten} turns to the case where the probability vector is replaced by a coarsened version, for which the analogous partially linear regression generally does not hold.

\section{Identification and its Spectral Geometry}
\label{sec:ident}

\subsection{The matrix moment identity}

\begin{theorem}[Identification]
\label{thm:ident}
Under Assumptions~\ref{ass:mean}, \ref{ass:cal} and \ref{ass:mom},
\begin{equation}
  \E[a R] \;=\; \Msf\,\tau.
  \label{eq:moment}
\end{equation}
If in addition Assumption~\ref{ass:rank} holds, then $\tau=\Msf^{-1}\E[aR]$ is point-identified from the joint law of $(Y,X,p)$.
\end{theorem}

The proof is in the supplement. Equation~\eqref{eq:moment} is the matrix analogue of the scalar identity $\E[(p-r(X))R]=\Vs\tau$ (the $K=1$ case, with $\Vs=\E[(p-r(X))^2]$): the scalar residual-score variance $\Vs$ becomes the covariance matrix $\Msf$, and the scalar covariance $\E[(p-r)R]$ becomes the cross-moment vector $\E[aR]$. Thus $\tau$ is the coefficient vector in the multivariate partial regression of $R$ on the residualized score $a$.

\subsection{Spectral characterization of identification failure}

Let $\Msf=\sum_{j=1}^{K}\lambda_j v_j v_j\T$ be the spectral decomposition, with $0\le\lambda_1\le\cdots\le\lambda_K$ and orthonormal eigenvectors $v_j$.

\begin{theorem}[Spectral identification]
\label{thm:spectral}
Under Assumptions~\ref{ass:mean}, \ref{ass:cal} and \ref{ass:mom}:
\begin{enumerate}[(a)]
\item The moment equation \eqref{eq:moment} identifies $\Msf\tau$, hence the
orthogonal projection of $\tau$ onto $\range(\Msf)=\mathrm{span}\{v_j:\lambda_j>0\}$.
The full vector $\tau$ is point-identified by \eqref{eq:moment} if and only if $\lambda_1>0$.
\item If $\rank(\Msf)=K-q<K$, the moment \eqref{eq:moment} does not
point-identify $\tau$: it determines only the projection
$P_{\range(\Msf)}\tau=\Msf^{+}\E[aR]$ (with $\Msf^{+}$ the Moore--Penrose
inverse), and its solution set is the affine subspace $\Msf^{+}\E[aR]+\nullsp(\Msf)$.
The components of $\tau$ in $\nullsp(\Msf)$ are unrestricted by the moment.
\item $\lambda_j=0$ if and only if $v_j\T p=v_j\T r(X)$ almost surely; that is,
the score combination $v_j\T p$ is a deterministic function of $X$.
\item The failure in (b) is genuinely observational, not an artifact of one
moment equation. Let $\tilde m(p,X)=\E[Y\mid p,X]$, $u=Y-\tilde m(p,X)$,
$g_0(u)=u/(1+|u|)$, and suppose there is $\underline\kappa>0$ with
$\kappa(p,X):=\E[u\,g_0(u)\mid p,X]\geq\underline\kappa$ a.s.\ and
$\|\tau\|_\infty\leq\underline\kappa/8$. Then for every
$z\in\nullsp(\Msf)$ with $\|z\|_\infty=1$ and every $c$ with
$|c|\leq\bar c:=\underline\kappa/8$, there exists a model satisfying
Assumptions~\ref{ass:mean}, \ref{ass:cal} and~\ref{ass:mom} with coefficient
vector $\tau+cz$ in which the observables $(Y,X,p)$ have exactly their
original joint distribution. The identified set is convex, contains the
segment $\{\tau+cz:|c|\leq\bar c\}$ along every null direction, and---whenever
each class probability is bounded away from zero on a positive-probability
set on which $\E[u^2\mid p,X]$ is essentially bounded---it is bounded, hence
a strict subset of the affine subspace in (b).
\end{enumerate}
\end{theorem}

Theorem~\ref{thm:spectral} (proved in the supplement) characterizes identification failure along the spectral directions of $\Msf$. Part~(c) says a direction collapses precisely when the corresponding linear combination of probabilities carries no residual variation beyond $X$. Part~(b) says the failure is directional: the projection of the effect onto $\range(\Msf)$ remains point-identified. Part~(d) upgrades the moment-level statement to observational equivalence, by tilting the conditional class probabilities with a bounded, conditionally mean-zero transform of the outcome residual. The construction and the resulting bounded segment are given in the supplement, together with the reasons the conditions are sufficient rather than sharp. When $K=1$ identification is all or nothing; with several classes an analyst can have sharp inference on some contrasts and no point identification on others in the same study, and the spectrum of $\widehat\Msf$ says which.

In the disparity-audit literature, part~(d) is the spectrum-indexed counterpart of the sharp partial-identification sets of \citet{kallus2022}, the difference being that Assumption~\ref{ass:mean} restricts the outcome--membership coupling throughout, so that only the collapsed directions remain set-identified.

\begin{remark}[Clustered eigenvalues]
\label{rem:cluster}
When eigenvalues cluster, the sample eigenvectors of $\widehat\Msf$ are unstable, and direction-specific quantities should be reported for the spanned eigenspace rather than a single eigenvector. In the simulations, inference on a fixed direction is unaffected even when two small eigenvalues nearly coincide.
\end{remark}

\section{Estimation and Inference}
\label{sec:inference}

\subsection{Estimator and Neyman orthogonality}

With nuisances $m,r$ estimated by cross-fitting, the moment is
\begin{equation}
  \psi(W;\tau,m,r)\;=\;a\,(R-a\T\tau),\qquad a=p-r(X),\ R=Y-m(X),
  \label{eq:psi}
\end{equation}
and the estimator solves $n^{-1}\sum_i\hat a_i(\hat R_i-\hat a_i\T\hat\tau)=0$, i.e.
\begin{equation}
  \hat\tau \;=\; \Bigl(n^{-1}\textstyle\sum_i \hat a_i\hat a_i\T\Bigr)^{-1}
                 \Bigl(n^{-1}\textstyle\sum_i \hat a_i\hat R_i\Bigr).
  \label{eq:est}
\end{equation}
We treat the score $p$ as an externally supplied calibrated output; the only estimated nuisances are $m$ and $r$. When $p$ is itself produced by a classifier trained in-sample, cross-fitting the classifier (as in Section~\ref{sec:adult}) controls overfitting, but our asymptotics condition on the score-generating mechanism.

\begin{proposition}[Full Neyman orthogonality]
\label{prop:ortho}
Under Assumptions~\ref{ass:mean}, \ref{ass:cal} and~\ref{ass:mom}, the moment \eqref{eq:psi} is Neyman-orthogonal with respect to both nuisances: the Gateaux derivatives of $\E[\psi]$ in the $m$- and $r$-directions vanish at the truth.
\end{proposition}

Proposition~\ref{prop:ortho} (proved in the supplement) makes \eqref{eq:est} compatible with machine-learning nuisance estimators under the usual second-order rate condition \citep{chernozhukov2018}. The symmetric residual form \eqref{eq:psi} is orthogonal in both $m$ and $r$.

\subsection{Multivariate sandwich CLT}

\begin{theorem}[Multivariate CLT]
\label{thm:clt}
Under Assumptions~\ref{ass:mean}--\ref{ass:mom}, i.i.d.\ sampling, and nuisance estimates satisfying $\|\hat m-m\|_{2}\,\|\hat r-r\|_{2}+\|\hat r-r\|_{2}^{2}=o_P(n^{-1/2})$ with cross-fitting,
\begin{equation}
  \sqrt{n}\,(\hat\tau-\tau)\;\dto\;\N\!\left(0,\;\Omega\right),
  \qquad \Omega=\Msf^{-1}\Sigma\,\Msf^{-1},
  \quad \Sigma=\E\!\left[u^2\,a a\T\right],\ u=R-a\T\tau.
  \label{eq:clt}
\end{equation}
Write $\widehat\Sigma=n^{-1}\sum_i\hat u_i^2\,\hat a_i\hat a_i\T$ with $\hat u_i=\hat R_i-\hat a_i\T\hat\tau$, $\widehat\Omega=\widehat\Msf^{-1}\widehat\Sigma\,\widehat\Msf^{-1}$, and the finite-sample variance estimate $\widehat V=\widehat\Omega/n$. Then $\widehat\Omega\pto\Omega$, Wald intervals $\hat\tau_k\pm z_{1-\alpha/2}\sqrt{\widehat V_{kk}}$ and the ellipsoid $(\hat\tau-\tau)\T\widehat V^{-1}(\hat\tau-\tau)\le\chi^2_{K,1-\alpha}$ have asymptotic coverage $1-\alpha$. Writing $\sigma^2_j=v_j\T\Sigma v_j$, the asymptotic variance along eigen-direction $j$ is $v_j\T\Omega v_j=\sigma_j^2/\lambda_j^2$, so precision degrades as $\lambda_j\to0$, quantifying Theorem~\ref{thm:spectral}.
\end{theorem}

The proof (in the supplement) is a standard $Z$-estimator argument using the orthogonality of Proposition~\ref{prop:ortho} to absorb the nuisance estimation error. The directional variance formula $v_j\T\Omega v_j=\sigma_j^2/\lambda_j^2$ links the spectral geometry to inference: a weakly-identified direction ($\lambda_j$ small) yields a wide confidence interval, while well-identified directions remain sharp. The exponent depends on how the moment noise behaves: with locally homoskedastic noise $\Sigma\approx\sigma^2\Msf$ giving $v_j\T\Omega v_j\approx\sigma^2/\lambda_j$, while at fixed $\Sigma$ it is $\lambda_j^{-2}$; the experiments show the realized scaling sits between $\lambda_j^{-1}$ and $\lambda_j^{-2}$ because $\Sigma$ co-varies with $\Msf$ across designs (Section~\ref{sec:sim}). A companion result in the supplement gives the direction-specific rate $\sqrt{n\lambda_{v,n}}$ at which a contrast is learned as its eigenvalue drifts to zero, and shows that the Wald interval remains correctly calibrated along such a direction under conditional homoskedasticity.

\subsection{Semiparametric efficiency}
\label{sec:eff}
\noindent By Theorem~\ref{thm:ident}, $\tau$ solves the conditional moment restriction $\E[\,R-a\T\tau\mid p,X\,]=0$, and the efficiency theory for such restrictions \citep{chamberlain1987} gives the semiparametric efficiency bound.

\begin{theorem}[Efficiency bound and optimal weighting]
\label{thm:eff}
Let $\sigma^{2}(p,X)=\Var(R\mid p,X)>0$ almost surely, let $\tilde p(X)=\E[\sigma^{-2}(p,X)\,p\mid X]/\E[\sigma^{-2}(p,X)\mid X]$ be the precision-weighted conditional mean of $p$, and set $\tilde a=p-\tilde p(X)$. Under Assumptions~\ref{ass:mean}--\ref{ass:mom}, assuming the displayed expectations are finite and $\E[\sigma^{-2}(p,X)\,\tilde a\tilde a\T]$ is nonsingular, the semiparametric efficiency bound for $\tau$ is
\begin{equation}
  V_{\mathrm{eff}}=\bigl(\E[\sigma^{-2}(p,X)\,\tilde a\tilde a\T]\bigr)^{-1},
  \label{eq:veff}
\end{equation}
attained by the estimator solving $\sum_i \sigma^{-2}(p_i,X_i)\,\tilde a_i\,(R_i-a_i\T\tau)=0$, with $\sigma^2$, $\tilde p$, $m$ and $r$ replaced by cross-fitted estimates. The unweighted estimator \eqref{eq:est} attains \eqref{eq:veff} if and only if $\sigma^{2}(p,X)\,a=C\tilde a$ almost surely for some nonsingular constant matrix $C$; a constant $\sigma^{2}$ is the leading sufficient case. In general $\Omega=\Msf^{-1}\Sigma\,\Msf^{-1}\succeq V_{\mathrm{eff}}$, with strict inequality along some direction under generic heteroskedasticity.
\end{theorem}

The proof (in the supplement) is a Chamberlain calculation carried out under the constraint that $\mu$ is unknown. An instrument $h(p,X)$ leaves the moment $\E[h\,(Y-\mu(X)-\tau\T p)]$ insensitive to perturbations of $\mu$ only if $\E[h\mid X]=0$, and since $\sigma^{2}$ depends on $p$, the naive choice $h=\sigma^{-2}a$ violates that condition; recentering $p$ at $\tilde p(X)$ restores it, and $\sigma^{-2}\tilde a$ is optimal among the instruments that satisfy it. Heteroskedasticity is intrinsic here: since $u=\tau\T(g-p)+\varepsilon$ and $\E[\varepsilon\mid G,p,X]=0$ kills the cross term, $\sigma^{2}(p,X)=\tau\T(\operatorname{diag}(p)-pp\T)\tau+\Var(\varepsilon\mid p,X)$, so the latent-membership term varies with $p$ even when $\Var(\varepsilon\mid p,X)$ is constant, and $\tilde p$ then departs from $r$. The efficient estimator therefore requires two nuisance functions beyond $(m,r)$, the conditional variance and the precision-weighted conditional mean. The unweighted estimator---used in the rest of the paper---needs neither and is orthogonal by Proposition~\ref{prop:ortho}.

\section{Coarsening the Probability Vector}
\label{sec:atten}

Equation~\eqref{eq:plr} describes the regression the analyst could run. It is common to run a different one. A common practice when probability vectors feed a downstream analysis is to discard them and keep a coarser object: an argmax one-hot label \citep{lee2013}, retained in some variants only when its predicted probability clears a confidence threshold \citep{sohn2020}, a top-$k$ summary, or a rounded version of $p$. This section characterizes the effect of that substitution.

\begin{definition}
A \emph{coarsening} is a measurable map $h=h(p,X)$ taking values in the same coordinates as membership, $h:\Delta_K\times\mathcal{X}\to\Delta_K$ or into the one-hot vertices. Write $a_h=h-\E[h\mid X]$ and let $\hat\tau_h$ solve $n^{-1}\sum_i\hat a_{h,i}(\hat R_i-\hat a_{h,i}\T\hat\tau_h)=0$.
\end{definition}

The requirement that $h$ live in the coordinates of $g$ is what makes the comparison with $\tau$ meaningful: a linear change of basis $h=Bp$ would return coefficients in transformed units rather than a distorted version of the same object. Coarsening in this sense is the operation studied under a different guise in the coarse-data literature \citep{heitjan1991}; here the coarsening is deliberate and its target is a downstream coefficient.

\begin{theorem}[Coarsening operator]
\label{thm:atten}
Under Assumptions~\ref{ass:mean}, \ref{ass:cal} and~\ref{ass:mom}, for any coarsening $h$ with $D_h=\E[a_ha_h\T]$ nonsingular,
\begin{equation}
  \plim\,\hat\tau_{h} \;=\; \mathcal{A}_h\,\tau,
  \qquad
  \mathcal{A}_h \;=\; D_h^{-1}C_h,\quad
  C_h=\E[a_h a\T].
  \label{eq:atten}
\end{equation}
\end{theorem}

The proof (in the supplement) has two steps: $\E[a_h\varepsilon]=0$ because $a_h$ is $\sigma(p,X)$-measurable, and conditioning on $(p,X)$ turns $\E[a_h(g-r(X))\T]$ into $\E[a_ha\T]$ by Assumption~\ref{ass:cal}. Nonsingularity of $D_h$ requires the coarsened label to retain residual variation given $X$; it fails, for instance, if the argmax never selects some class.

Writing $u=a-a_h$ for the signal the coarsening discards gives the form we use throughout.

\begin{corollary}[Distortion is covariance with the discarded signal]
\label{cor:decomp}
$\mathcal{A}_h = I + D_h^{-1}\E[a_h u\T]$. Hence $\mathcal{A}_h=I$ if and only if $\E[a_hu\T]=0$, and $\mathcal{A}_h=c\,I$ for a scalar $c$ if and only if $\E[a_hu\T]=(c-1)D_h$.
\end{corollary}

Coarsening leaves the coefficient undistorted exactly when what it throws away is uncorrelated with what it keeps. That condition is restrictive. The second part of the corollary says that even proportional shrinkage is non-generic: it requires a covariance condition that the designs and audits below never satisfy.

\begin{corollary}[Directional retention]
\label{cor:aniso}
Let $s_j=v_j\T \mathcal{A}_h v_j$ be the retention coefficient along the $j$th eigen-direction of $\Msf$. Coarsening acts uniformly on $\tau$ only when $\mathcal{A}_h=cI$, equivalently when $\E[a_hu\T]=(c-1)D_h$; outside that case it is direction-specific, and the $s_j$ need neither coincide across $j$ nor lie in $[0,1]$.
\end{corollary}

Coarsening therefore need not act on the effect vector uniformly. Empirically, in every design and audit we examine, the $s_j$ of a nontrivial coarsening lie in $(0,1)$ and differ across directions, so coarsening attenuates some contrasts much more than others. The corollary does not assert this in general, and the diagnostic below reports the $s_j$ rather than presuming their range or their ordering.

$D_h$ and $C_h$ are functionals of the observed $(p,X)$ alone: $\mathcal{A}_h$ can be estimated without labels and without invoking Assumption~\ref{ass:cal}, which enters only in reading $\mathcal{A}_h$ as the distortion of $\tau$. When $K=1$, \eqref{eq:atten} reduces to the scalar projection slope $\kappa=\E[a_{h}a]/\E[a_{h}^2]$, the familiar dilution factor of a binary pseudo-label. Under a calibrated probability and monotone thresholding it falls in $(0,1)$ \citeparenKbs; the operator result does not guarantee this in general. \citet{chen2019} analyze the same choice for a binary group-mean disparity and decompose the bias of the thresholded estimator, which for that estimand may inflate as well as dilute; the target here is a structural coefficient, and \eqref{eq:atten} collects the whole distortion into the single matrix $\mathcal{A}_h$. With several classes the distortion is directional rather than a single scalar factor.

\section{Coarsening and the Validity of Downstream Inference}
\label{sec:coverage}

Theorem~\ref{thm:atten} concerns the point estimate. We turn to the interval an analyst reports alongside it, which is what enters a decision.

Consider the analyst who computes $\hat\tau_h$, forms the sandwich as though the moment were correctly specified,
\begin{equation}
  \widehat\Omega_h=\widehat D_h^{-1}\widehat\Sigma_h\widehat D_h^{-1},
  \qquad
  \widehat\Sigma_h=n^{-1}\textstyle\sum_i \hat u_{h,i}^{2}\,\hat a_{h,i}\hat a_{h,i}\T,
  \quad \hat u_h=\hat R-\hat a_h\T\hat\tau_h,
  \label{eq:naive-sandwich}
\end{equation}
and reports $v\T\hat\tau_h\pm z_{1-\alpha/2}\sqrt{v\T\widehat\Omega_h v/n}$ for a direction $v$. Write the directional coarsening bias and variance as $b_v=v\T(\mathcal{A}_h-I)\tau$ and $\sigma^2_{h,v}=v\T\Omega_h v$, and define the standardized bias
\[
  \nu_{v,n}=\frac{\sqrt n\,b_v}{\sigma_{h,v}}.
\]

\begin{proposition}[Coverage under coarsening]
\label{prop:coverage}
Under the conditions of Theorem~\ref{thm:clt} applied to the moment $a_h(R-a_h\T\vartheta)$, whose population solution is $\mathcal{A}_h\tau$, along sequences with $b_v=b/\sqrt n$ for fixed $b$,
\begin{equation}
  \Pbb\bigl(v\T\tau \in \mathrm{CI}\bigr)\;\longrightarrow\;
  \Phi(z_{1-\alpha/2}-\nu_v)-\Phi(-z_{1-\alpha/2}-\nu_v),
  \qquad \nu_v=b/\sigma_{h,v}.
  \label{eq:coverage}
\end{equation}
For fixed $b_v\neq0$, $|\nu_{v,n}|\to\infty$ and the coverage tends to zero.
\end{proposition}

The proof (in the supplement) is the studentized form of Theorem~\ref{thm:clt} with the probability limit displaced from $\tau$ to $\mathcal{A}_h\tau$. The formula itself is the power function of a $z$-test.

\begin{remark}[Greater apparent precision can worsen coverage]
\label{rem:precise}
The magnitude $|\nu_{v,n}|=\sqrt n\,|b_v|/\sigma_{h,v}$ increases in $1/\sigma_{h,v}$. For a given bias, coverage therefore degrades as the reported standard error shrinks. Coarsened labels are more extreme than the probabilities they replace, and the reported variance $v\T\Omega_hv$, with $\Omega_h=D_h^{-1}\Sigma_hD_h^{-1}$, can be smaller than $v\T\Omega v$. It is smaller along the reported direction for every nontrivial coarsening in Table~\ref{tab:e11}: the reported interval is narrower than the uncoarsened one, around a displaced center.
\end{remark}

Because $\mathcal{A}_h$ and $\Omega_h$ are identified from the observed data, only $b_v$ additionally requires $\tau$, and the uncoarsened estimator supplies it. The asymptotic approximation to its coverage can therefore be computed before the interval is reported:
\begin{equation}
  \hat\nu_v=\frac{v\T(\widehat{\mathcal{A}}_h-I)\hat\tau}{\sqrt{v\T\widehat\Omega_h v/n}},
  \qquad
  \widehat{\mathrm{cov}}=\Phi(z_{1-\alpha/2}-\hat\nu_v)-\Phi(-z_{1-\alpha/2}-\hat\nu_v).
  \label{eq:cov-hat}
\end{equation}

\citet{battaglia2024} study machine-learning-generated regressors and show that treating them as observed data can yield biased estimates and invalid inference. There the distortion comes from estimation error in the generated regressor and correction requires auxiliary information. Here the probability vector is observed exactly and the distortion is a deterministic function of it, which is why \eqref{eq:cov-hat} needs no auxiliary information. Prediction-powered inference \citep{angelopoulos2023} restores validity using gold labels on part of the analysis sample. The diagnostic above applies when no such labels exist and assesses coverage without correcting the estimator. \citet{li2026} requires no assumptions on the upstream classifier, treating the proxy as a link to an auxiliary validation sample and bounding the parameter by optimal transport. That paper also observes that a retained probability vector is more informative than a categorical label, which is the statement \eqref{eq:atten} makes quantitative---again without a validation sample.

\section{Simulation Studies}
\label{sec:sim}

We report here the three experiments most central to the theory, labeled E6, E11 and E12 in a numbering shared with the supplement. The data-generating process has $X\in\R^3$ standard normal and $W=(X,S)$, where $S\in\R^2$ is an auxiliary signal outside the downstream control set. The probability vector $p$ is drawn from a Dirichlet centred at $\mathrm{softmax}(B_r\T W)$, with dispersion $\sigma_u$ controlling $\Var(p\mid X)$ and hence the spectrum of $\Msf$, and the latent class is $G\mid p\sim\mathrm{Categorical}(p)$ on $\{1,\dots,K+1\}$ with $K=3$. The outcome is $Y=X\T b_m+g\T\tau+\N(0,1)$ with $\tau=(1.0,0.7,-0.4)\T$. Nuisances $m,r$ are cross-fitted (5 folds) by degree-2 polynomial regression. Monte Carlo replicate counts are stated with each experiment.

\subsection{Inference for the uncoarsened estimator (E6)}
Componentwise and joint Wald coverage for $\hat\tau$ itself lie in $[0.938,0.960]$ and $[0.946,0.958]$ at nominal $95\%$ over $1000$ replicates as $n$ grows from $500$ to $5000$, consistent with Theorem~\ref{thm:clt}. The table and the normal QQ-plot can be found in the supplement. Under the same data-generating process, Experiment E11 below gives coverage of $0.95$ without coarsening and $0.01$ after argmax coarsening along the same contrast.

\subsection{The coarsening operator (E11)}
Five coarsenings of the same probability vector are compared at $n=8000$: the identity, rounding $p$ to a $0.20$ grid and renormalizing, a top-2 renormalization, a confidence threshold at $0.60$ below which the unit receives the zero vector rather than being dropped, and the argmax one-hot. For each we estimate $\mathcal{A}_h$ from $(p,X)$ alone, run the coarsened estimator, and record the interval it reports along the weakest eigen-direction $v_1$ of $\widehat\Msf$. For none of the four nontrivial coarsenings in Table~\ref{tab:e11} is the operator a multiple of the identity: its diagonal falls as the coarsening discards more, and its off-diagonal entries are non-zero throughout, so the distortion is not proportional across contrasts. The reported interval narrows as the coarsening becomes more severe, while its coverage falls from nominal to almost nothing. The argmax interval is $39\%$ shorter than the one the uncoarsened regression would have produced and covers $1\%$ of the time.

\begin{table}[t]
\centering
\spacingset{1}
\caption{Experiment E11: five coarsenings of the same probability vector ($n=8000$, $K=3$,
$400$ replicates). $\mathcal{A}_h$ is estimated from $(p,X)$ alone. SE and coverage
refer to the nominal-$95\%$ Wald interval along the weakest eigen-direction of
$\widehat\Msf$.}
\label{tab:e11}
\begin{threeparttable}
\begin{tabular}{lccrr}
\toprule
Coarsening $h$ & $\operatorname{diag}(\mathcal{A}_h)$ & $\max|\text{off-diag}|$ &
mean SE & coverage\\
\midrule
identity ($h=p$)        & $1.00\ \ 1.00\ \ 1.00$ & $0.000$ & $0.0405$ & $0.952$\\
round to $0.20$ grid    & $0.89\ \ 0.89\ \ 0.82$ & $0.029$ & $0.0377$ & $0.910$\\
top-2 renormalized      & $0.73\ \ 0.75\ \ 0.57$ & $0.059$ & $0.0336$ & $0.675$\\
threshold $0.60$        & $0.47\ \ 0.48\ \ 0.35$ & $0.192$ & $0.0332$ & $0.533$\\
argmax one-hot          & $0.41\ \ 0.42\ \ 0.30$ & $0.096$ & $0.0246$ & $0.010$\\
\bottomrule
\end{tabular}
\end{threeparttable}
\end{table}

\subsection{The coverage formula is calibrated (E12)}
Proposition~\ref{prop:coverage} predicts coverage from the noncentrality index $\nu_v$. Grouping $600$ replicates at $n=4000$ by $|\nu_v|$ and comparing realized with predicted coverage within groups, the formula tracks the realized value to within $0.02$ across the whole range, from nominal coverage down to $0.003$. Substituting $\hat\tau$ for $\tau$ in $\hat\nu_v$ widens these gaps, most where the distortion is most severe, so \eqref{eq:cov-hat} is best read as an accurate ordering and an approximate level. A sweep over the dispersion of $p$ and over the overlap between the classifier's features and the controls leaves the ordering unchanged. In every configuration examined the argmax interval covers between $1\%$ and $6\%$. The grouped comparison and the sweep are tabulated in the supplement.

The supplement reports four further experiments. Argmax coarsening acts anisotropically, with the retention coefficient increasing in the eigenvalue (Corollary~\ref{cor:aniso}). Setting $K=1$ recovers the scalar $(\Vs)^{-2}$ sandwich law. The efficient instrument of Theorem~\ref{thm:eff} cuts the variance threefold, while the naive weighting that omits its recentering reports standard errors that are too small. While $n\lambda_{v,n}\to\infty$, inference remains valid at the $\sqrt{n\lambda_{v,n}}$ rate along a collapsing direction.

\section{A Disparity Audit with a Proxied Protected Attribute}
\label{sec:adult}

Auditing outcome disparities across groups when the protected attribute is unobserved and must be proxied \citep{kallus2022,chen2019} is naturally multi-class. We use the standard training file of UCI Adult \citep{kohavi1996} ($n=30{,}162$ after listwise deletion). The latent group is race in $K+1=4$ classes (reference \textit{White}; classes \textit{Black}, \textit{Asian--Pacific}, \textit{Other}, with the original \textit{Amer-Indian-Eskimo} category included in \textit{Other}). The controls $X$ are age, education, and sex. The proxy is a multinomial classifier on a richer feature set $W\supsetneq X$ (occupation, workclass, relationship, native region, capital variables), cross-fitted over $5$ folds to an out-of-fold probability vector $p$. A controlled validation on this real score spectrum, in which the estimator recovers a known $\tau$ along every eigen-direction, is reported in the supplement---here we run the genuine audit.

The outcome $Y$ is the high-income indicator, the attribute is treated as unobserved, and the estimator uses only the proxy $p$. Because race is in fact recorded here, we also compute the full-information benchmark from the true labels. The hard estimator uses the argmax class. There is no guarantee the assumptions hold---an off-the-shelf race proxy need not be conditionally calibrated, and its features may carry direct income channels.

The estimated spectrum of $\widehat\Msf$ is anisotropic, with eigenvalues $(0.0006,0.0069,0.0177)$ and condition number $31$: one disparity direction carries substantially more residual variation than the other two. Table~\ref{tab:adult} reports, in the eigenbasis of $\widehat\Msf$, the retention coefficient under argmax coarsening, the estimate from the proxy alone, and the full-information benchmark. The retention coefficients range from $s_1=0.46$ to $s_3=0.90$. On real data, coarsening removes more than half the signal in one direction and about a tenth in another. The identity behind Corollary~\ref{cor:aniso} requires no calibration assumption, since $\widehat{\mathcal{A}}_h$ is computed from $(p,X)$ alone.

The proxy-based estimate matches the benchmark in the best-identified direction ($0.034\pm0.017$ against $0.028$). In the two weak directions it departs sharply, as eigenvalues an order of magnitude smaller make the estimator more sensitive to violations of the identifying assumptions. The supplement reports clean calibration batteries but a direct-channel check giving $\Delta R^2=0.135$, which points to the exclusion component rather than calibration. We report this as a cautionary instance, not a substantive claim about disparities.

\begin{table}[t]
\centering
\spacingset{1}
\caption{UCI Adult disparity audit, in the eigenbasis of $\widehat\Msf$
(weakest to strongest identified direction). $s_j$: retention coefficient under
hard-labeling; soft: proxy-based estimate of $v_j\T\tau$ (sandwich SE);
benchmark: same projection from the observed labels.}
\label{tab:adult}
\begin{threeparttable}
\begin{tabular}{lrrrr}
\toprule
Direction & $\lambda_j(\widehat\Msf)$ & $s_j$ &
soft $v_j\T\hat\tau$ (SE) & benchmark $v_j\T\hat\tau$\\
\midrule
$v_1$ (weakest)   & 0.0006 & 0.46 & $-0.584\ (0.080)$ & $-0.079$\\
$v_2$             & 0.0069 & 0.72 & $-0.472\ (0.025)$ & $-0.065$\\
$v_3$ (strongest) & 0.0177 & 0.90 & $\phantom{-}0.034\ (0.017)$ & $\phantom{-}0.028$\\
\bottomrule
\end{tabular}
\end{threeparttable}
\end{table}

\section{Diagnosing Assumption Failure by Recalibration}
\label{sec:recal}

The framework rests on two substantive assumptions: conditional calibration (Assumption~\ref{ass:cal}) and the exclusion component of Assumption~\ref{ass:mean}---the requirement that the score carry no direct outcome channel given $(G,X)$. In both the Adult audit above and a real-data land-cover audit on UCI Forest CoverType (a distance-to-water outcome with a random-forest cover-type classifier; details in the supplement), at least one of these assumptions is violated. A practitioner needs to know which assumption is more problematic, because recalibration on a labeled subsample can address calibration failure but not an exclusion violation.

With a labeled subsample the two assumptions can be probed separately. Split the labeled data in half. On the first half, fit the recalibration map $\tilde p=\widehat{\E}[g\mid p,X]$ (here, a multinomial logistic regression on the log-odds of $p$, the controls, and their interactions). On the held-out half, re-run the audit with $\tilde p$ in place of $p$. Recalibration targets Assumption~\ref{ass:cal}. Were the population map known exactly, $\E[g\mid\tilde p,X]=\tilde p$ would hold by the tower property, and the out-of-sample calibration battery assesses how closely the estimated map achieves this. It does not target the exclusion component. If $p$ carries a direct outcome channel, so may any function of $(p,X)$. If Assumption~\ref{ass:mean} holds for $(p,X)$ it also holds for $(\tilde p,X)$, since $\sigma(\tilde p,X)\subseteq\sigma(p,X)$. Taken together with the calibration battery and the direct-channel check, the change in the audit under recalibration is therefore informative about which failure dominates.

If the structural mean carries a direct channel, $\E[Y\mid G,p,X]=\mu(X)+\tau\T g+d(p,X)$ for a direct channel $d$, then $\E[aR]=\Msf\tau+\E[a\,d]$ and $\plim\hat\tau=\tau+\Msf^{-1}\E[a\,d]$. In the eigenbasis of $\Msf$, the direct-channel term is scaled by the inverse eigenvalues, making weak directions particularly sensitive to the violation. In our audits recalibration compresses the probability vector toward its conditional expectation and shrinks the spectrum, without in general eliminating a direct outcome channel. Recalibration can therefore amplify bias from an exclusion violation.

\begin{proposition}[Sharp sensitivity to miscalibration]
\label{prop:cal}
Let Assumptions~\ref{ass:mean}, \ref{ass:rank} and~\ref{ass:mom} hold, and let $\E[g\mid p,X]=p+\eta(p,X)$ replace Assumption~\ref{ass:cal}, with $\mathcal{H}_\delta$ denoting the class of errors with $\|\eta\|_\infty\le\delta$ almost surely for which $p+\eta$ remains a valid class-probability vector. Then (a) $\plim\hat\tau=\tau+\Msf^{-1}N\tau$ with $N=\E[a\,\eta\T]$; the bias is linear in $\tau$, so if $\tau=0$ miscalibration alone cannot produce a nonzero coefficient vector, although a single null component can still be distorted by the mixing. (b) For every unit vector $v$,
\begin{equation}
  \sup_{\eta\in\mathcal{H}_\delta}
  \bigl|v\T(\plim\hat\tau-\tau)\bigr|
  \;\le\;\delta\,\|\tau\|_1\,\E\bigl|v\T\Msf^{-1}a\bigr|,
  \label{eq:cal-bound}
\end{equation}
with equality whenever every class probability, including the reference, is at least $K\delta$ almost surely. Along the eigen-direction $v_j$ the bound equals $\delta\|\tau\|_1\E|v_j\T a|/\lambda_j\le\delta\|\tau\|_1\lambda_j^{-1/2}$.
\end{proposition}

Proposition~\ref{prop:cal} (proved in the supplement) is the quantitative form of the warning after Assumption~\ref{ass:cal}: the sharp bound is $\delta\|\tau\|_1\,\E|v_j\T\Msf^{-1}a|$ and is no larger than $\delta\|\tau\|_1\lambda_j^{-1/2}$, so the envelope over directions widens as $\lambda_j^{-1/2}$ and weakly identified directions can be the most fragile to calibration drift. The sharp bound itself need not be monotone in the eigenvalue, since $\E|v_j\T\Msf^{-1}a|$ varies with the direction.

The proposition specializes the approximate-moment-condition framework of \citet{armstrong2021} to this model. In that framework, a moment restriction is allowed to fail within a budget, and the resulting identified set and confidence intervals are derived in general. Here the bound has a closed form, is attained when every class probability is at least $K\delta$, and provides an eigenvalue-indexed envelope for the sensitivity of different contrasts. For inference on the resulting set rather than the point, that general machinery applies directly. The eigenvalue-indexed reading also informs the protocol below: recalibration can reduce the calibration error represented by $\delta$ but compresses the spectrum, so its net effect is direction-specific.

Table~S5 in the supplement reports the experiment on both audits, together with an independent check of the exclusion component (the partial $R^2$ of the classifier features $W$ for the outcome given the true class and $X$). On Adult the rich-basis battery finds essentially no conditional miscalibration ($R^2\le0.006$), and recalibration accordingly changes nothing. The direct-channel check gives $\Delta R^2=0.135$: the classifier's features retain predictive power for earnings conditional on true class and $X$. The evidence therefore points to an exclusion violation in the weak directions, which recalibration does not address. CoverType differs. The rich basis reveals severe conditional miscalibration ($R^2$ up to $0.32$) and recalibration reduces it out of sample to $R^2\le0.003$, yet the audit deteriorates. The recalibrated spectrum compresses (largest eigenvalue $0.122\to0.029$), and the resulting estimate is consistent with amplification through the direct-channel term carried by the soil and wilderness indicators ($\Delta R^2=0.081$), as described by the bias expression above.

A practical diagnostic sequence is to run the rich-basis $(p,X)$ calibration battery on the labeled subsample, since an $X$-only check can miss violations that live in the $p$-direction, and then recalibrate and re-run out of sample. If the gap to an available benchmark closes, the evidence points to calibration as the dominant failure, which is repairable. If the audit barely moves despite a clean battery, or moves sharply after recalibration, that is evidence of a remaining exclusion problem, and the soft audit should not be reported without the direct-channel check. The experiment requires no labeled data beyond the subsample the battery already uses.
\section{A Voter-Turnout Disparity Audit}
\label{sec:bisg}

In the Adult and CoverType audits, the diagnostics indicate a material direct channel from the classifier's features to the outcome. More generally, those features may carry outcome-relevant information beyond the latent class in an observational regression. The assumptions are more plausible under a different measurement structure, one in which the classifier reads a proxy that is approximately a pure indicator of the latent class, with little or no direct channel to the outcome. We next consider a real-world setting in which the diagnostics are substantially more favorable, though conditional calibration is still not exact.

We audit racial disparities in voter turnout when the protected attribute is proxied rather than observed---the setting of the surname- and geography-based race-proxy literature \citep{elliott2009,imai2016}. The data are the public North Carolina voter file and its voter-history file for Mecklenburg County ($n=464{,}700$ active registrants after matching), together with the 2010 Census surname list \citep{comenetz2016}. The latent group $G$ is self-reported race/ethnicity in $K+1=4$ classes (reference \emph{White}; \emph{Black}, \emph{Hispanic}, \emph{Asian}), recorded on the voter file and used only to form the benchmark. The probability vector $p$ is the surname list's national $\Pbb(\text{race}\mid\text{surname})$, an off-the-shelf classifier we take as given. The controls $X$ are the electoral precinct. The outcome $Y$ is the number of November general elections in which the registrant voted, 2016--2024. Unlike Bayesian Improved Surname Geocoding, which folds geography into the proxy itself, we keep the proxy surname-only and let precinct enter separately as a control. The framework needs residual variation in every score direction once the controls are partialled out, and a proxy already conditioned on geography would leave little of it after precinct is removed.

Surname is plausibly closer to a pure indicator of race, as it carries genuine, moderate information about race. Any direct channel to turnout once race and place are fixed is hard to motivate substantively. The data are consistent with that reading, though they cannot establish it. The direct-channel check of Section~\ref{sec:recal} is essentially zero, as the partial $R^2$ of the surname score for turnout given true race and precinct is $0.001$. The surname classifier is genuinely uncertain but skilful (argmax accuracy $0.63$ against a White base rate of $0.55$; many surnames are racially ambiguous), so the score is neither degenerate nor deterministic. The residual-score covariance $\widehat\Msf$ is well conditioned (eigenvalues $0.015,0.024,0.061$; condition number $4.1$). No direction appears close to collapse.

Table~\ref{tab:bisg} reports the estimated turnout gaps. The full-information benchmark (from self-reported race) shows Hispanic and Asian registrants turning out roughly half a standard deviation below White registrants net of precinct, and Black registrants modestly below. Hard-labeling attenuates these gaps toward zero in this application, matching the estimated coarsening operator of Theorem~\ref{thm:atten}, and does most damage where the benchmark gap is largest, shrinking the Hispanic and Asian gaps by roughly a fifth ($-0.454$ and $-0.428$ against benchmark values $-0.551$ and $-0.525$). The soft estimator, built from the surname proxy alone, comes close on both ($-0.520$ and $-0.524$). Its RMSE relative to the benchmark is $0.042$, compared with $0.082$ for hard-labeling, and this difference is stable across resamplings.

\begin{table}[t]
\centering
\spacingset{1}
\caption{Census-surname voter-turnout audit, North Carolina (Mecklenburg County,
$n=464{,}700$), in the original class basis. Benchmark uses self-reported
race; soft and hard use only the surname proxy. Entries are turnout gaps
relative to White, net of precinct (outcome standardized).}
\label{tab:bisg}
\begin{threeparttable}
\begin{tabular}{lrrrr}
\toprule
Group & benchmark (true race) & soft (surname) & hard (argmax) & hard/benchmark\\
\midrule
Black    & $-0.076$ & $-0.011$ & $-0.040$ & $0.53$\\
Hispanic & $-0.551$ & $-0.520$ & $-0.454$ & $0.82$\\
Asian    & $-0.525$ & $-0.524$ & $-0.428$ & $0.82$\\
\bottomrule
\end{tabular}
\begin{tablenotes}\small
\item \textit{Notes}: RMSE versus the benchmark: soft $0.042$, hard $0.082$.
Direct-channel check (exclusion) $\Delta R^2=0.001$; surname-argmax accuracy
$0.63$ (White base rate $0.55$); $\widehat\Msf$ condition number $4.1$. The last column is the ratio of the hard-label estimate to the benchmark, in the original class basis; it is not the eigen-direction retention coefficient $s_j$.
\end{tablenotes}
\end{threeparttable}
\end{table}

At $n\approx465{,}000$ the uncoarsened estimator does not match the benchmark exactly. Its small deviations are statistically resolvable, and on the already-small Black gap the residual bias makes it no better than the coarsened one. This reflects that a national surname list is not exactly conditionally calibrated to a single county. It nevertheless matches the benchmark closely on the two large gaps that coarsening shrinks. The gain is also not an artifact of recalibration. The raw Census probabilities already outperform the hard labels against the benchmark, and recalibrating them locally on a labeled half leaves the conclusion unchanged.

\section{Discussion}
\label{sec:disc}

The covariance matrix $\Msf$ of the partialled-out probability vector determines which contrasts the regression \eqref{eq:plr} recovers and, together with the moment noise, how precisely they are estimated. The coarsening operator $\mathcal{A}_h$ determines the displacement induced by replacing that vector with a coarser version. Together with the reported standard error, that displacement determines coverage through the noncentrality index. In the binary case both reduce to scalars. There, the operator is a dilution factor and identification is all or nothing, whereas with several classes each becomes direction-specific.

\subsection*{Scope and operating conditions}

The three real-data audits illustrate the conditions under which the framework is informative. It is most informative when all of the following hold:

\begin{enumerate}[(i)]
\item Membership in several mutually exclusive latent classes is unobserved, and the probability vector is conditionally calibrated. Conditional calibration holds by construction for a correctly specified posterior. With a labeled subsample, it can be assessed and improved by recalibration (Section~\ref{sec:recal}). The CoverType audit shows that improving calibration need not improve the downstream audit when other identifying assumptions fail.

\item The residual-score covariance $\Msf=\E[\Cov(p\mid X)]$ is full rank, as required by Assumption~\ref{ass:rank}. The score must retain variation in every direction after the controls are partialled out. Both sides of this trade-off matter. Omitting controls required by the outcome equation makes the estimate unreliable in the usual way, and the no-controls estimator in the land-cover audit (supplement) sign-flips. At the other extreme, with $X$ as rich as $W$, the residual $a$ degenerates and the spectrum collapses.

\item The classifier is genuinely uncertain. When it is nearly deterministic, the probability vector concentrates on the vertices of the simplex, so coarsened and uncoarsened labels nearly coincide and $\mathcal{A}_h\to I$. The framework remains valid in that case but adds little. Hard-but-learnable tasks are settings in which coarsening can be particularly consequential.

\item There is no trusted forward model of the measurement process. Where a generative likelihood for the probability vector is available and believed, full-likelihood methods may be more efficient.
\end{enumerate}

Potential applications include multi-group fairness audits with validation labels, hard medical subtyping against registry gold standards, and verbal-autopsy cause-of-death attribution. The surname-proxy turnout audit of Section~\ref{sec:bisg} comes closest to satisfying all four. Land-cover and similar remote-sensing problems qualify only conditionally. The favorable case uses a proxy that is closer to a pure indicator (a surname carries information about race with little evidence of a direct channel to turnout), whereas the conditional cases use covariates that also drive the outcome. Where a labeled subsample exists, the recalibration experiment of Section~\ref{sec:recal} provides an empirical check on the calibration assumption and evidence about which assumption is more problematic. The exclusion restriction cannot be established from these diagnostics and remains a substantive judgment.

The efficient estimator attains the bound but needs the two extra nuisances of Section~\ref{sec:eff}. The supplement shows that the gain survives a well-specified variance model, while a misspecified one can substantially reduce it. Behavior under fully nonparametric estimation of both is unresolved. Inference on the set induced by a calibration budget can be carried out with the machinery of \citet{armstrong2021}, which is built for moment conditions that fail by a bounded amount. The identified set along a collapsed direction is a different problem. By Theorem~\ref{thm:spectral}(d) it does not shrink with the sample size. Finally, the coarsening operator could be used to build a direction-by-direction selection rule, calibrated against a label-based benchmark on a validation sample and then applied without labels to choose between the uncoarsened and coarsened estimators. For a binary latent class \citetKbs\ develops such a rule, indexing the attenuation factor by the confidence threshold. Carrying it over to several classes is left to future work.

\section*{Funding}
The author reports no external funding for this work.

\section*{Supplementary Material}
Proofs of all results stated in the main text, two further identification results, additional simulation experiments, a land-cover audit, and full reports of the empirical audits.

\section*{Competing Interests}
The author declares no competing interests.

\section*{Data Availability Statement}
All data used in this article are publicly available. The UCI Adult (\href{https://doi.org/10.24432/C5XW20}{doi:10.24432/C5XW20}) and Forest CoverType (\href{https://doi.org/10.24432/C50K5N}{doi:10.24432/C50K5N}) data are available from the UCI Machine Learning Repository; the 2010 Census surname list is available from the U.S.\ Census Bureau (\url{https://www.census.gov/topics/population/genealogy/data/2010_surnames.html}); and the North Carolina voter registration and voter-history files are available from the North Carolina State Board of Elections (\url{https://www.ncsbe.gov/results-data/voter-registration-data}).

\section*{Declaration of Generative AI Use}

The author used Claude Opus 4.8 (Anthropic) for assistance with manuscript editing and code development. All outputs were reviewed and verified by the author, who takes full responsibility for the content of the article.


%% file: supplement-body.tex
\noindent
This supplement contains proofs of all results stated in the main text, together with two additional identification results (Section~\ref{sec:addid}), further simulation experiments, and a land-cover audit. Throughout, $a=p-r(X)$, $R=Y-m(X)$, $u=R-a\T\tau$, and $\Msf=\E[aa\T]=\sum_j\lambda_j v_jv_j\T$ is the residual-score covariance matrix with eigenpairs $(\lambda_j,v_j)$, $\lambda_1\le\cdots\le\lambda_K$. References to numbered assumptions and results are to the main text.

\section{Proof of Lemma \MainLemResid (residual decomposition)}
By Assumption~1, $\E[Y\mid G,p,X]=\mu(X)+\tau\T g$. Taking $\E[\cdot\mid X]$ and using $\E[g\mid X]=\E[\E[g\mid p,X]\mid X]=\E[p\mid X]=r(X)$ (Assumption~2) gives $m(X)=\mu(X)+\tau\T r(X)$. Hence $R=(Y-\E[Y\mid G,p,X])+(\E[Y\mid G,p,X]-m(X))=\varepsilon+\tau\T(g-r(X))$, where $\varepsilon:=Y-\E[Y\mid G,p,X]$ satisfies $\E[\varepsilon\mid G,p,X]=0$. \hfill$\square$

\section{Proof of Theorem \MainThmIdent (identification)}
By Lemma~\MainLemResid, $\E[aR]=\E[a\,\tau\T(g-r(X))]+\E[a\varepsilon]$. Since $a$ is $\sigma(p,X)$-measurable and $\E[\varepsilon\mid G,p,X]=0$, $\E[a\varepsilon]=\E[a\,\E[\varepsilon\mid p,X]]=0$. For the first term, $\E[a(g-r(X))\T]=\E[a\,\E[(g-r(X))\T\mid p,X]]=\E[a(p-r(X))\T]=\E[aa\T]=\Msf$, the middle equality being Assumption~2. Hence $\E[aR]=\Msf\tau$; under Assumption~3, $\Msf$ is invertible and $\tau=\Msf^{-1}\E[aR]$. \hfill$\square$

\noindent The result rests on Assumption~2, which gives $\E[g-p\mid p,X]=0$: this is the Berkson-type orthogonality used here, placing the error orthogonal to the regressor rather than inside it, so no attenuation factor survives. Theorem~\MainThmAtten records what replaces that step once $p$ is coarsened.

\section{Proof of Theorem \MainThmSpectral (spectral characterization)}

\begin{enumerate}[(a)]
\item The moment $\E[aR]=\Msf\tau$ determines $\Msf\tau$ from observables. Writing $\tau=\sum_j(v_j\T\tau)v_j$ gives $\Msf\tau=\sum_{j:\lambda_j>0}\lambda_j(v_j\T\tau)v_j$, which fixes $v_j\T\tau$ for each $j$ with $\lambda_j>0$ and nothing about the components with $\lambda_j=0$; hence $\tau$ is identified iff $\lambda_1>0$.

\item The solution set of $\Msf\tau=\E[aR]$ is $\Msf^{+}\E[aR]+\nullsp(\Msf)$, with $\Msf^{+}\E[aR]=P_{\range(\Msf)}\tau$ the minimum-norm solution; the moment restricts only the projection of $\tau$ onto $\range(\Msf)$.

\item $\lambda_j=v_j\T\Msf v_j=\E[(v_j\T a)^2]\ge0$, which is zero iff $v_j\T a=0$ a.s., i.e.\ $v_j\T p=v_j\T r(X)$ a.s.

\item Let $\tilde m(p,X)=\E[Y\mid p,X]=\mu(X)+\tau\T p$ (Assumptions~1--2), $u=Y-\tilde m(p,X)$, $g_0(u)=u/(1+|u|)$, and $\kappa(p,X)=\E[u\,g_0(u)\mid p,X]\ge\underline\kappa>0$; note $u\,g_0(u)\ge0$ and $|g_0|<1$. For $z\in\nullsp(\Msf)$, part~(c) gives $z\T p=z\T r(X)$ a.s., so $z\T p$ is $\sigma(X)$-measurable. Fix such a $z$ with $\|z\|_\infty=1$ and $|c|\le\bar c=\underline\kappa/8$; set $\tau'=\tau+cz$, $\mu'(X)=\mu(X)-c\,z\T r(X)$ (which is $X$-measurable), and
\begin{gather*}
  \psi=g_0(u)-\E[g_0(u)\mid p,X],\qquad
  c_k=\frac{p_k(\tau'_k-\tau'^{\top}p)}{\kappa(p,X)},\\
  q_k=p_k+c_k\psi\ (k\le K),\qquad q_{K+1}=1-\textstyle\sum_{k\le K}q_k.
\end{gather*}
Since $|\psi|<2$ and, assuming $\|\tau\|_\infty\le\underline\kappa/8$, $|\tau'_k-\tau'^{\top}p|\le2\|\tau\|_\infty+2|c|\le\underline\kappa/2$, we get $|c_k\psi|<p_k$, so $q_k\ge0$; also $\sum_{k\le K}c_k=(\tau'^{\top}p)p_{K+1}/\kappa$ with $|\tau'^{\top}p|\le\underline\kappa/4$, so $q_{K+1}>0$. Thus $q$ is a valid distribution on $\{1,\dots,K+1\}$.
\end{enumerate}

Adjoin $V\sim\mathrm{Uniform}(0,1)$ independent of $(Y,X,p)$ and draw $g'$ categorically from $q$, so $\Pbb(g'=e_k\mid Y,p,X)=q_k$; the observables are untouched, so observational equivalence is exact. Calibration holds because $\E[g'_k\mid p,X]=\E[q_k\mid p,X]=p_k$ (as $\E[\psi\mid p,X]=0$). For the structural mean, $\E[Y\psi\mid p,X]=\kappa(p,X)$, so on $\{p_k>0\}$, $\E[Y\mid g'=e_k,p,X]=\tilde m+c_k\kappa/p_k=\tilde m+\tau'_k-\tau'^{\top}p =\mu'(X)+\tau'_k$, and $\E[Y\mid g'=e_{K+1},p,X]=\mu'(X)$; hence $\E[Y\mid g',p,X]=\mu'(X)+\tau'^{\top}g'$, which is Assumption~1 with coefficient $\tau'$. Thus every $\tau+cz$ with $|c|\le\bar c$ is observationally indistinguishable, giving a nondegenerate segment in each null direction.

Convexity: mixing two such kernels $q',q''$ preserves calibration and mixes the branch means linearly, so the coefficient set is convex. Boundedness: for any valid model with coefficient $\tau^{\circ}$, Cauchy--Schwarz with $(q^{\circ}_k)^2\le q^{\circ}_k$ gives, on $\{p_k\ge\underline p\}$, $|\tau^{\circ}_k-\tau^{\circ\top}p|\le(\E[u^2\mid p,X])^{1/2}/p_k^{1/2} \le\bar\sigma/\sqrt{\underline p}$; since $\tau^{\circ}\mapsto (\tau^{\circ}_k-\tau^{\circ\top}p)_k$ is invertible in the simplex interior ($I_K-\mathbf1 p\T$ has eigenvalues $1$ and $p_{K+1}$), $\tau^{\circ}$ is confined to a bounded set. \hfill$\square$

\section{Proof of Proposition \MainPropOrtho (Neyman orthogonality)}
Let $\psi=(p-r)(Y-m-(p-r)\T\tau)$. For $m\mapsto m+t\delta_m$, $\tfrac{d}{dt}\E[\psi]|_0=-\E[(p-r)\delta_m]=-\E[\E[a\mid X]\delta_m]=0$. For $r\mapsto r+t\delta_r$, the derivative is $-\E[\delta_r(Y-m-a\T\tau)]+\E[a\,\delta_r\T\tau] =-\E[\delta_r\E[u\mid X]]+\E[\E[a\mid X]\delta_r\T\tau]=0$, since $\E[u\mid X]=0$ and $\E[a\mid X]=0$. \hfill$\square$

\section{Proof of Theorem \MainThmClt (multivariate central limit theorem)}
The estimator $\hat\tau=(n^{-1}\sum_i\hat a_i\hat a_i\T)^{-1} (n^{-1}\sum_i\hat a_i\hat R_i)$ solves the moment $\psi=a(R-a\T\tau)$. By Proposition~\MainPropOrtho and cross-fitting the empirical-process contribution of the nuisance estimates is negligible, and the population remainder is $O_P(\|\hat m-m\|_2\|\hat r-r\|_2+\|\hat r-r\|_2^2)$: writing $\hat a=a-\Delta_r$ and $\hat R=R-\Delta_m$ with $\Delta_r=\hat r-r$, $\Delta_m=\hat m-m$, the cross terms $\E[a\Delta_m]$ and $\E[\Delta_r R]$ vanish because $\E[a\mid X]=0$ and $\E[R\mid X]=0$, the remaining numerator term is $\E[\Delta_r\Delta_m]$, and the denominator satisfies $\E[\hat a\hat a\T]=\Msf+\E[\Delta_r\Delta_r\T]$, which contributes the squared term. Under the rate condition of Theorem~\MainThmClt both are $o_P(n^{-1/2})$; see \citet{chernozhukov2018} for the general orthogonal-score and cross-fitting argument. Hence $\sqrt n(\hat\tau-\tau)=\Msf^{-1}n^{-1/2}\sum_i a_iu_i+o_P(1)$ with $\E[a_iu_i]=0$. The central limit theorem and Slutsky give $\sqrt n(\hat\tau-\tau)\dto\N(0,\Omega)$, $\Omega=\Msf^{-1}\Sigma\Msf^{-1}$, $\Sigma=\E[u^2aa\T]$; consistency of $\widehat\Omega$ follows from the law of large numbers and continuity of inversion \citep{vandervaart1998}. Finally $v_j\T\Omega v_j=\lambda_j^{-2}v_j\T\Sigma v_j$ since $\Msf^{-1}v_j=\lambda_j^{-1}v_j$. \hfill$\square$

\section[Proof of the direction-specific rate and collapse-robust inference]{Proof of Proposition~\ref{prop:weakid} (direction-specific rate and collapse-robust inference)}
\label{sec:weakid-proof}
\noindent From the influence representation of Theorem~\MainThmClt, $v\T(\hat\tau-\tau)=\lambda_{v,n}^{-1}n^{-1}\sum_i(v\T a_i)u_i+o_P(\cdot)$ for $v$ the eigenvector with eigenvalue $\lambda_{v,n}$. The summand has variance $v\T\Sigma v=\sigma^2\lambda_{v,n}(1+o(1))$, so by the Lindeberg CLT $v\T(\hat\tau-\tau)=O_P((n\lambda_{v,n})^{-1/2})$ and $\sqrt{n\lambda_{v,n}}\,v\T(\hat\tau-\tau)\dto\N(0,\sigma^2)$, giving~(i). For (ii), $v\T\widehat V v=\sigma^2/(n\lambda_{v,n})(1+o_P(1))$; the factor $\lambda_{v,n}^{-1}$ cancels in the studentized ratio, so it is $\N(0,1)$ whenever $n\lambda_{v,n}\to\infty$. In the oracle Gaussian model the numerator $n^{-1/2}\sum_i(v\T a_i)u_i$ is, conditionally on $\{a_i\}$, exactly $\N(0,\sigma^2 n^{-1}\sum_i(v\T a_i)^2)$ and the denominator concentrates at rate $n^{-1/2}$ for any $\lambda_{v,n}>0$, so if $n\lambda_{v,n}\to\kappa$ the studentized statistic remains $\N(0,1)$ while $v\T\hat\tau$ is inconsistent; coverage stays nominal with $O(1)$ half-width. \hfill$\square$

\section{Proof of Theorem \MainThmEff (efficiency bound and optimal weighting)}
\label{sec:proof-eff}
By Theorem~\MainThmIdent the model is $Y=\mu(X)+\tau\T p+u$ with $\E[u\mid p,X]=0$ and $\Var(u\mid p,X)=\sigma^2(p,X)$, where $\mu$ is an unrestricted function of $X$; write $Z=(p,X)$ and note $Y-\mu(X)-\tau\T p=R-a\T\tau=u$.

\emph{Admissible instruments.} Consider moments $\E[h(Z)\,(Y-\mu(X)-\tau\T p)]$ with $\E\|h\|^2<\infty$. Along the path $\mu\mapsto\mu+t\delta$ the derivative at $t=0$ is $-\E[h\,\delta(X)]=-\E[\E[h\mid X]\,\delta(X)]$, so the moment is insensitive to the unknown $\mu$ in every direction $\delta\in L_2(X)$ if and only if $\E[h\mid X]=0$ a.s. Let $\mathcal{H}$ collect the $h$ satisfying this with $\E[hp\T]$ nonsingular. These are exactly the moment functions orthogonal to the scores generated by perturbing $\mu$, so the minimum of the asymptotic variance over $\mathcal{H}$ is the semiparametric efficiency bound for $\tau$ in the model with $\mu$ unknown \citep{chamberlain1987,newey1994}. For $h\in\mathcal{H}$ the $Z$-estimator solving $\sum_ih_i(Y_i-\hat\mu(X_i)-\tau\T p_i)=0$ satisfies $\sqrt n(\hat\tau-\tau)\dto\N(0,V(h))$ with $V(h)=(\E[hp\T])^{-1}\E[\sigma^2hh\T](\E[ph\T])^{-1}$, the nuisance estimation contributing no first-order term under the regularity conditions for preliminary nonparametric estimators in \citet{newey1994}.

\emph{The optimum.} Because $\E[h\mid X]=0$, $\E[hp\T]=\E[h\tilde a\T]$ for $\tilde a=p-\tilde p(X)$ and any $X$-measurable $\tilde p$ for which the required moments exist. Take $\tilde p(X)=\E[\sigma^{-2}p\mid X]/\E[\sigma^{-2}\mid X]$, the unique choice for which $h^{*}=\sigma^{-2}\tilde a$ is itself admissible, since $\E[\sigma^{-2}\tilde a\mid X]=\E[\sigma^{-2}p\mid X]-\tilde p\,\E[\sigma^{-2}\mid X]=0$; write $J=\E[\sigma^{-2}\tilde a\tilde a\T]$. For any $h\in\mathcal{H}$ the matrix Cauchy--Schwarz inequality applied to $\sigma h$ and $\sigma^{-1}\tilde a$ gives $\E[\sigma^2hh\T]\succeq\E[h\tilde a\T]J^{-1}\E[\tilde ah\T]$, the difference being $\E[ww\T]\succeq0$ for $w=\sigma h-\sigma^{-1}\E[h\tilde a\T]J^{-1}\tilde a$. Pre- and post-multiplying by $(\E[h\tilde a\T])^{-1}$ and its transpose yields $V(h)\succeq J^{-1}$, with equality if and only if $w=0$ a.s., that is $h=C\sigma^{-2}\tilde a$ with $C=\E[h\tilde a\T]J^{-1}$ nonsingular. Both $\E[h^{*}p\T]$ and $\E[\sigma^2h^{*}h^{*\top}]$ equal $J$, so $V(h^{*})=J^{-1}=V_{\mathrm{eff}}$ and the bound is attained.

\emph{The unweighted estimator.} $h=a$ is admissible because $\E[a\mid X]=0$. Since $\E[ap\T]=\Msf$ and $\E[\sigma^2aa\T]=\Sigma$, we have $V=\Msf^{-1}\Sigma\Msf^{-1}\succeq V_{\mathrm{eff}}$, with equality iff $\sigma^2(p,X)\,a=C\tilde a$ a.s. A constant $\sigma^2$ gives $\tilde p=r$, hence $\tilde a=a$ and $C=\sigma^2I$; heteroskedasticity in $X$ alone also gives $\tilde a=a$, but then equality requires $\sigma^2(X)$ to be constant wherever $a\neq0$. The naive instrument $\sigma^{-2}a$ is generally inadmissible: $\E[\sigma^{-2}a\mid X]=\E[\sigma^{-2}\mid X]\,(\tilde p(X)-r(X))$, which vanishes only when $\Cov(\sigma^{-2}(p,X),\,p\mid X)=0$. That fails generically here. Since $u=\tau\T(g-p)+\varepsilon$ with $\E[g-p\mid p,X]=0$, $\Var(g\mid p,X)=\operatorname{diag}(p)-pp\T$ and $\Cov(\tau\T(g-p),\varepsilon\mid p,X)=0$, we obtain $\sigma^2(p,X)=\tau\T(\operatorname{diag}(p)-pp\T)\tau+\Var(\varepsilon\mid p,X)$, whose latent-membership component generally varies with $p$ when $\tau\neq0$. \hfill$\square$

\section{Proof of Theorem \MainThmAtten (the coarsening operator)}
Let $h$ be any coarsening and $a_h=h-\E[h\mid X]$, so that $\hat\tau_h\pto D_h^{-1}\E[a_hR]$ with $D_h=\E[a_ha_h\T]$ assumed nonsingular. By Lemma~\MainLemResid, $R=\tau\T(g-r(X))+\varepsilon$ with $\E[\varepsilon\mid G,p,X]=0$. Since $a_h$ is $\sigma(p,X)$-measurable, $\E[a_h\varepsilon]=\E[a_h\,\E[\varepsilon\mid p,X]]=0$. For the remaining term, conditioning on $(p,X)$ and applying Assumption~2 gives
\begin{equation*}
  \E[a_h(g-r(X))\T]=\E\bigl[a_h\,\E[(g-r(X))\T\mid p,X]\bigr]=\E[a_h(p-r(X))\T]=\E[a_ha\T]=C_h.
\end{equation*}
Hence $\E[a_hR]=C_h\tau$ and $\plim\hat\tau_h=D_h^{-1}C_h\tau=\mathcal{A}_h\tau$. Writing $u=a-a_h$ for the discarded signal, $C_h=\E[a_h(a_h+u)\T]=D_h+\E[a_hu\T]$, so $\mathcal{A}_h=I+D_h^{-1}\E[a_hu\T]$, which is the identity in the statement. It equals $I$ if and only if $\E[a_hu\T]=0$. Nothing in the argument uses the form of $h$, so the conclusion holds for every coarsening. The argmax one-hot label is the special case $h=\tilde g$. Nonsingularity of $D_h$ requires the coarsened vector to retain residual variation given $X$, which can fail if the coarsening never selects some class.
\hfill$\square$

\section{Proof of Proposition \MainPropCal (sharp sensitivity to miscalibration)}
\begin{enumerate}[(a)]

\item With $\varepsilon=Y-\E[Y\mid G,p,X]$, $\E[a\varepsilon]=0$ since $a$ is $\sigma(p,X)$-measurable. By Assumption~1 and the tower property, $m(X)=\mu(X)+\tau\T\E[g\mid X]=\mu(X)+\tau\T(r(X)+\bar\eta(X))$ with $\bar\eta(X)=\E[\eta\mid X]$. Hence $\E[aR]=\E[a\,g\T]\tau-\E[a\,(r(X)+\bar\eta(X))\T]\tau$. The second term vanishes because $\E[a\mid X]=0$, and for the first, $\E[ag\T]=\E[a\,\E[g\mid p,X]\T]=\E[a(p+\eta)\T]=\Msf+N$, using $\E[ap\T]=\E[aa\T]+\E[a\,r(X)\T]=\Msf$. Thus $\E[aR]=(\Msf+N)\tau$. With nuisances converging to the observable regressions $m,r$ as in Theorem~\MainThmClt, $\hat\tau\pto\Msf^{-1}(\Msf+N)\tau=\tau+\Msf^{-1}N\tau$, linear in $\tau$.

\item $v\T\Msf^{-1}N\tau=\E[(v\T\Msf^{-1}a)(\eta\T\tau)]$, and for $\eta\in\mathcal{H}_\delta$, $|\eta\T\tau|\le\|\eta\|_\infty\|\tau\|_1 \le\delta\|\tau\|_1$ pointwise, whence the bound. For attainment set $\eta^{*}=-\delta\,\sgn(v\T\Msf^{-1}a)\,s$ with $s_k=\sgn(\tau_k)$ (any sign when $\tau_k=0$). Then $\eta^{*\top}\tau=-\delta\,\sgn(v\T\Msf^{-1}a)\|\tau\|_1$, so $\E[(v\T\Msf^{-1}a)\,\eta^{*\top}\tau]=-\delta\|\tau\|_1\E|v\T\Msf^{-1}a|$ attains the bound in magnitude. When every class probability is at least $K\delta$ a.s., $\eta^{*}$ is admissible. We have $\|\eta^{*}\|_\infty\le\delta$. For $k\le K$, $p_k+\eta^{*}_k\ge K\delta-\delta\ge0$ and $p_k+\eta^{*}_k\le p_k+\delta\le(1-p_{K+1})+\delta\le1$ since $p_{K+1}\ge K\delta\ge\delta$. The implied reference probability $p_{K+1}-\1\T\eta^{*}$ lies in $[p_{K+1}-K\delta,\,p_{K+1}+K\delta] \subseteq[0,1]$: the lower end follows from $p_{K+1}\ge K\delta$, while the upper end follows from $\sum_{k\le K}p_k\ge K\cdot K\delta\ge K\delta$, which gives $p_{K+1}\le1-K\delta$. Finally $v_j\T\Msf^{-1}a=\lambda_j^{-1}v_j\T a$ and, by Cauchy--Schwarz, $\E|v_j\T a|\le(\E[(v_j\T a)^2])^{1/2}=\lambda_j^{1/2}$.
\end{enumerate}
\hfill$\square$

\section{Additional Identification Results}
\label{sec:addid}

\subsection{The marginal contrast is also identified}
\label{sec:marginal}
The coefficient $\tau_k$ is the within-$X$ contrast between class $k$ and the reference. A natural companion estimand is the marginal contrast $\Delta_k=\E[Y\mid G=k]-\E[Y\mid G=K{+}1]$, which does not hold $X$ fixed. Under conditional calibration the two differ by a compositional term that is itself point-identified---in contrast with fully latent measurement-error settings, where the covariate composition of the latent classes is lost.

\begin{proposition}[Marginal contrast]
\label{prop:marginal}
Let Assumptions~1--4 hold and $\Pbb(G=k)>0$ for all $k$. Then, with $\mu(X)=m(X)-\tau\T r(X)$ and $p_{K+1}=1-\1\T p$,
\[
\Delta_k=\tau_k+C_k,\qquad C_k=\frac{\E[\mu(X)\,p_k]}{\E[p_k]}-\frac{\E[\mu(X)\,p_{K+1}]}{\E[p_{K+1}]},
\]
so the vector $\Delta$ is point-identified alongside $\tau$. Moreover $\Delta_k=\tau_k$ if and only if $\E[\mu(X)\mid G=k]=\E[\mu(X)\mid G=K{+}1]$, i.e.\ classes $k$ and the reference are balanced in $\mu(X)$.
\end{proposition}

\emph{Proof.} For $k\le K$, one-hotness gives $(\tau\T g)g_k=\tau_k g_k$, so by Assumption~1 and the tower property $\E[Yg_k]=\E[\mu(X)g_k]+\tau_k\E[g_k]$. By Assumption~2 and the tower property, $\E[\mu(X)g_k]=\E[\mu(X)\,\E[g_k\mid p,X]]=\E[\mu(X)p_k]$ and $\E[g_k]=\E[p_k]=\Pbb(G=k)$, so $\E[Y\mid G=k]=\tau_k+\E[\mu(X)p_k]/\E[p_k]$. The same computation with $g_{K+1}=1-\1\T g$ (on which $\tau\T g\,g_{K+1}=0$) gives $\E[Y\mid G=K{+}1]=\E[\mu(X)p_{K+1}]/\E[p_{K+1}]$. Taking differences gives the display. The identity $\mu=m-\tau\T r$ follows from Lemma~\MainLemResid's decomposition $m=\mu+\tau\T r$, and the balance equivalence is immediate since $\E[\mu(X)\mid G=k]=\E[\mu(X)p_k]/\E[p_k]$. \hfill$\square$

The same tower argument identifies the full latent-class covariate law, $\Pbb(X\in A\mid G=k)=\E[\mathbf{1}_A(X)\,p_k]/\E[p_k]$. Under conditional calibration, reporting $\tau$ or $\Delta$ is a choice of estimand, not an identification issue, and both are estimable by cross-fitted plug-in.

\subsection{Heterogeneous effects and the matrix-weighted estimand}
\label{sec:hetero}

\begin{proposition}[Matrix-weighted estimand]
\label{prop:hetero}
Replace the constant coefficient in Assumption~1 by $\E[Y\mid G,p,X]=\mu(X)+\tau(X)\T g$ with $\E\|\tau(X)\|^2<\infty$, and keep Assumptions~2 and~4. Let $\Msf(X)=\Cov(p\mid X)$, so $\Msf=\E[\Msf(X)]$, and let $\Msf$ be nonsingular. Then $\E[aR]=\E[\Msf(X)\,\tau(X)]$ and
\[
\plim\hat\tau \;=\;\bar\tau \;:=\;\Msf^{-1}\,\E\bigl[\Msf(X)\,\tau(X)\bigr],
\]
the $\Msf(X)$-weighted average of $\tau(\cdot)$; when $K=1$ this is $\E[\tau(X)\Var(p\mid X)]/\E[\Var(p\mid X)]$, and $\bar\tau=\tau$ under constancy.
\end{proposition}

\emph{Proof.} As in Lemma~\MainLemResid, $m(X)=\mu(X)+\tau(X)\T r(X)$ and $R=\tau(X)\T(g-r(X))+\varepsilon$ with $\E[\varepsilon\mid G,p,X]=0$, so $\E[aR]=\E[\E[a(g-r(X))\T\mid X]\,\tau(X)]$. Conditionally on $X$, $\E[a(g-r(X))\T\mid X]=\E[a\,g\T\mid X]=\E[a\,p\T\mid X] =\E[a\,a\T\mid X]=\Msf(X)$, using $\E[a\mid X]=0$ and Assumption~2. The plug-in limit follows as in Theorem~\MainThmClt. \hfill$\square$

The weight $\Msf(X)$ is the conditional residual-score covariance at $X$: covariate cells in which the score varies more along a direction contribute more to $\bar\tau$ along that direction. This is the multiclass analogue of the variance-weighting familiar from partially linear models under effect heterogeneity, and the estimand should be interpreted accordingly when constancy is implausible.

\section{Additional Simulation Experiments}
The following experiments supplement those in the main text.

\subsection{Scalar reduction recovers the binary law (E3, E4)}
Setting $K=1$ recovers the binary model. Across a $\sigma_u$ sweep the raw log--log slope of $\Var(\hat\tau)$ on $\Vs$ is $-1.02$, not $-2$, because the moment-noise variance moves with $\Vs$. Once normalized by the moment noise the slope is exactly $-2.00$ ($R^2=1.00$), recovering the scalar $(\Vs)^{-2}$ sandwich law of the binary case (Figure~\ref{fig:e3}), and the sandwich interval attains coverage in $[0.953,0.963]$ across two decades of $\Vs$ (Figure~\ref{fig:e4}). This clarifies the scaling: the apparent departure from $-2$ in the multivariate variance scaling is a $\Sigma$-effect, not a failure of the sandwich, and the clean $-2$ re-emerges once the moment noise is held fixed.

\begin{figure}[t]
\centering
\includegraphics[width=\textwidth]{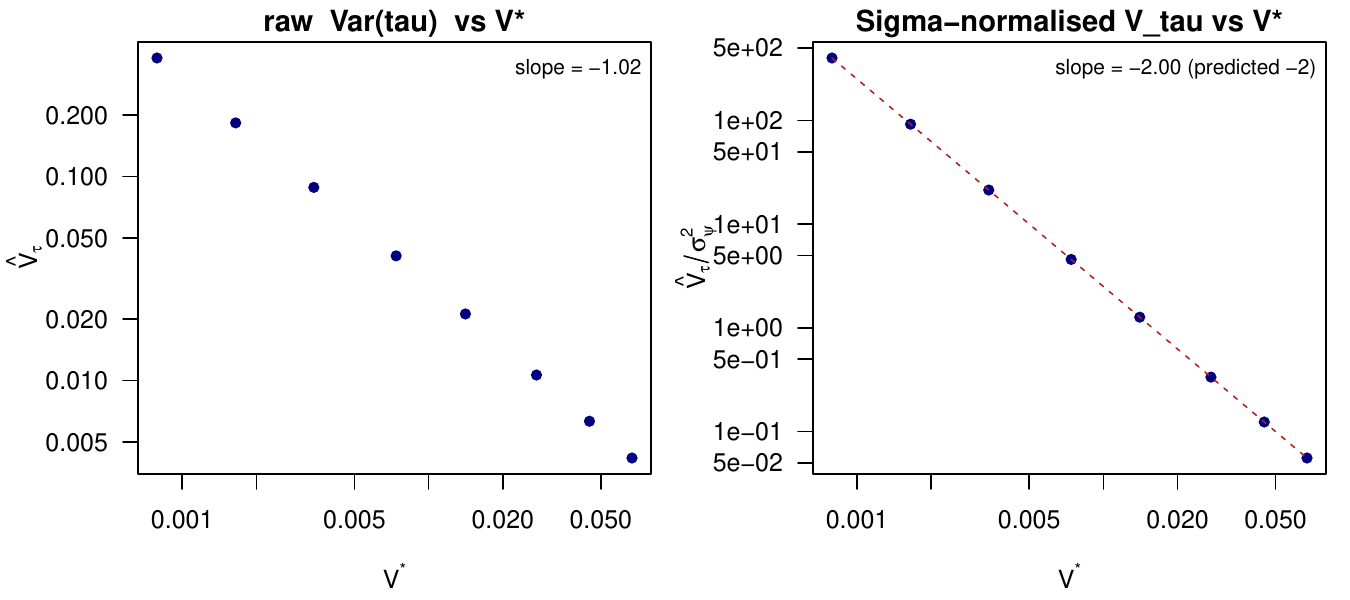}
\caption{E3. Scalar reduction ($K=1$): raw variance against $\Vs$ (left,
slope $-1.02$) and moment-noise-normalized variance against $\Vs$ (right,
slope $-2.00$, the predicted $(\Vs)^{-2}$ law).}
\label{fig:e3}
\end{figure}

\begin{figure}[t]
\centering
\includegraphics[width=0.69\textwidth]{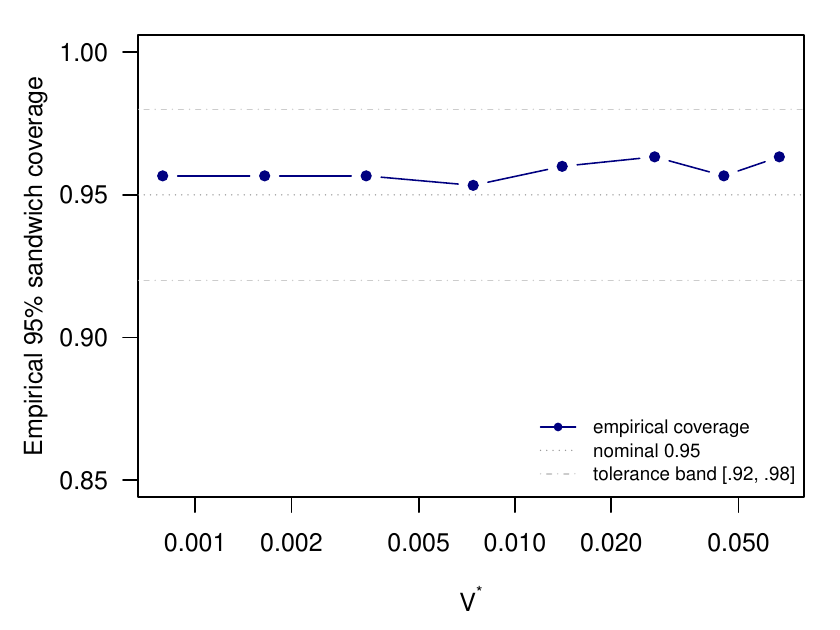}
\caption{E4. Empirical coverage of the nominal-$95\%$ sandwich interval in
the scalar reduction, across two decades of $\Vs$.}
\label{fig:e4}
\end{figure}

\subsection{Optimal weighting attains the efficiency bound (E7)}
\label{sec:e7}
We examine Theorem~\MainThmEff under conditional heteroskedasticity, generating the outcome noise with standard deviation $0.6\,e^{0.8X_1}$. Table~\ref{tab:e7} compares three instruments: the unweighted $h=a$ of (\MainEqEst), the naive $h=\sigma^{-2}a$, and the efficient $h=\sigma^{-2}\tilde a$ of (\MainEqVeff), the last two using the true $\sigma^{2}(p,X)$ so that the comparison isolates the instrument. The unweighted estimator carries about three times the variance of the bound (mean ratio $3.10$). The naive instrument captures most but not all of that gain, and falls further behind as $\sigma_u$ grows and $\sigma^2$ becomes more responsive to $p$ (ratios $1.08$, $1.11$, $1.22$).

Inadmissibility also affects inference, not only precision. The naive instrument's sandwich reduces to $(\E[\sigma^{-2}aa\T])^{-1}$, which lies below the bound because $\tilde p$ is the $\sigma^{-2}$-weighted projection of $p$. The reported standard errors therefore understate the true variability, by $5\%$ at $\sigma_u=0.20$ and $14\%$ at $\sigma_u=0.45$, and Wald coverage drifts from $0.945$ to $0.936$. The efficient instrument shows no such gap: its sandwich matches its empirical variance to within $4\%$ and coverage stays at $0.95$.

Replacing the oracle $\sigma^2$ by a cross-fitted estimate (log-linear in $(X,p)$, matching the multiplicative form of the noise) leaves the efficient estimator at the bound, with variance ratios $1.06$, $1.02$ and $0.98$ and coverage $0.948$--$0.954$. Thus, in this design, the two extra nuisances can be estimated without a detectable first-order loss. This conclusion is design-specific: a variance model fitted linearly rather than log-linearly on the same data inflates the ratio to about $1.6$.

\begin{table}[t]
\centering
\spacingset{1}
\caption{E7: trace of $\Var(\hat\tau)$ for the unweighted, naively weighted
and efficient instruments under heteroskedasticity ($n=4000$, $500$ replicates,
true $\sigma^2$). The last column is the ratio of the naive estimator's mean
sandwich to its empirical variance.}
\label{tab:e7}
\begin{threeparttable}
\begin{tabular}{rrrrrrr}
\toprule
$\sigma_u$ & $h=a$ & $h=\sigma^{-2}a$ & $h=\sigma^{-2}\tilde a$ &
$\tfrac{\text{unw}}{\text{eff}}$ & $\tfrac{\text{naive}}{\text{eff}}$ &
naive $\tfrac{\text{sand}}{\text{emp}}$\\
\midrule
0.20 & 0.149 & 0.051 & 0.047 & 3.18 & 1.08 & 0.95\\
0.30 & 0.066 & 0.026 & 0.024 & 2.81 & 1.11 & 0.93\\
0.45 & 0.038 & 0.014 & 0.011 & 3.32 & 1.22 & 0.86\\
\bottomrule
\end{tabular}
\end{threeparttable}
\end{table}

\subsection{Collapse-robust inference and the $\sqrt{n\lambda}$ rate (E10)}
We examine Proposition~\ref{prop:weakid} numerically on the oracle linear-Gaussian model with a single weak eigen-direction driven to zero at the bounded-information rate $\lambda_{v,n}=5/n$ (so $n\lambda_{v,n}=5$ is fixed). Table~\ref{tab:e10} reports the outcome. The standard deviation of the weak-direction estimate does not shrink with $n$ (it sits at $0.44$ from $n=10^3$ to $n=6.4\times10^4$)---the estimator is inconsistent along the collapsing direction. Yet the rescaled deviation $\sqrt{n\lambda_{v,n}}\,\mathrm{SD}$ is constant at $\approx\!1=\sigma$, matching the same $\sqrt{n\lambda}$ scaling as in part~(i), and the Wald coverage stays at the nominal $0.95$ throughout. Part~(ii) does not cover this case, since it assumes $n\lambda_{v,n}\to\infty$---the oracle-Gaussian statement that follows it does. Across all rates we examined ($\lambda$ fixed, $c/\sqrt n$, $c/n$) the studentized statistic was indistinguishable from $\N(0,1)$ (excess kurtosis within $\pm0.2$). A near-degenerate weak pair ($\lambda_1\approx\lambda_2=4/n$) behaved identically. In this design the intervals widen as identification weakens without losing their nominal coverage.

\begin{table}[t]
\centering
\spacingset{1}
\caption{E10: weak direction driven to zero at $\lambda_{v,n}=5/n$ (oracle
model, $\sigma=1$, $1{,}500$ replicates). The point estimate is inconsistent
(SD flat in $n$) but the Wald interval is valid (coverage $0.95$), and the
$\sqrt{n\lambda}$-rescaled SD is constant.}
\label{tab:e10}
\begin{threeparttable}
\begin{tabular}{rrrrr}
\toprule
$n$ & $\lambda_{v,n}$ & $\mathrm{SD}(v\T\hat\tau)$ &
$\sqrt{n\lambda_{v,n}}\,\mathrm{SD}$ & Wald coverage\\
\midrule
1{,}000  & $5.0\times10^{-3}$ & 0.442 & 0.99 & 0.951\\
4{,}000  & $1.3\times10^{-3}$ & 0.448 & 1.00 & 0.951\\
16{,}000 & $3.1\times10^{-4}$ & 0.445 & 1.00 & 0.959\\
64{,}000 & $7.8\times10^{-5}$ & 0.439 & 0.98 & 0.959\\
\bottomrule
\end{tabular}
\end{threeparttable}
\end{table}


\subsection{Hard-labeling attenuates anisotropically (E5)}
We estimate the attenuation matrix $\mathcal{A}$ of Theorem~\MainThmAtten and read off the per-direction retention coefficient $s_j=v_j\T \mathcal{A} v_j$. Table~\ref{tab:e5} shows that the best-identified direction retains $0.41$--$0.54$ of its signal while the weakest retains only $0.12$--$0.40$. The median anisotropy ratio $s_3/s_1$ is $2.03$, and $s_j$ increases monotonically with $\lambda_j$ (pooled slope of $s_j$ on $\log\lambda_j$ positive, $R^2=0.82$). In this design hard-labeling erodes the weakly-identified contrasts roughly twice as fast as the well-identified ones, one instance of the directional heterogeneity of Corollary~\MainCorAniso. Figure~\ref{fig:e5} displays the relationship.

\begin{table}[t]
\centering
\spacingset{1}
\caption{E5: retention coefficient $s_j=v_j\T \mathcal{A} v_j$ by eigen-direction and score
dispersion ($n=4000$, $K=3$). Direction~1 is the weakest-identified, 3 the
strongest.}
\label{tab:e5}
\begin{threeparttable}
\begin{tabular}{rrrrrrr}
\toprule
$\sigma_u$ & $\lambda_1$ & $\lambda_3$ & $s_1$ (weak) & $s_2$ & $s_3$ (strong) &
mean $s$\\
\midrule
0.12 & 0.0016 & 0.072 & 0.116 & 0.236 & 0.411 & 0.254\\
0.20 & 0.0034 & 0.078 & 0.180 & 0.258 & 0.429 & 0.289\\
0.30 & 0.0066 & 0.088 & 0.266 & 0.306 & 0.461 & 0.345\\
0.45 & 0.0131 & 0.111 & 0.399 & 0.413 & 0.536 & 0.450\\
\bottomrule
\end{tabular}
\end{threeparttable}
\end{table}

\begin{figure}[t]
\centering
\includegraphics[width=\textwidth]{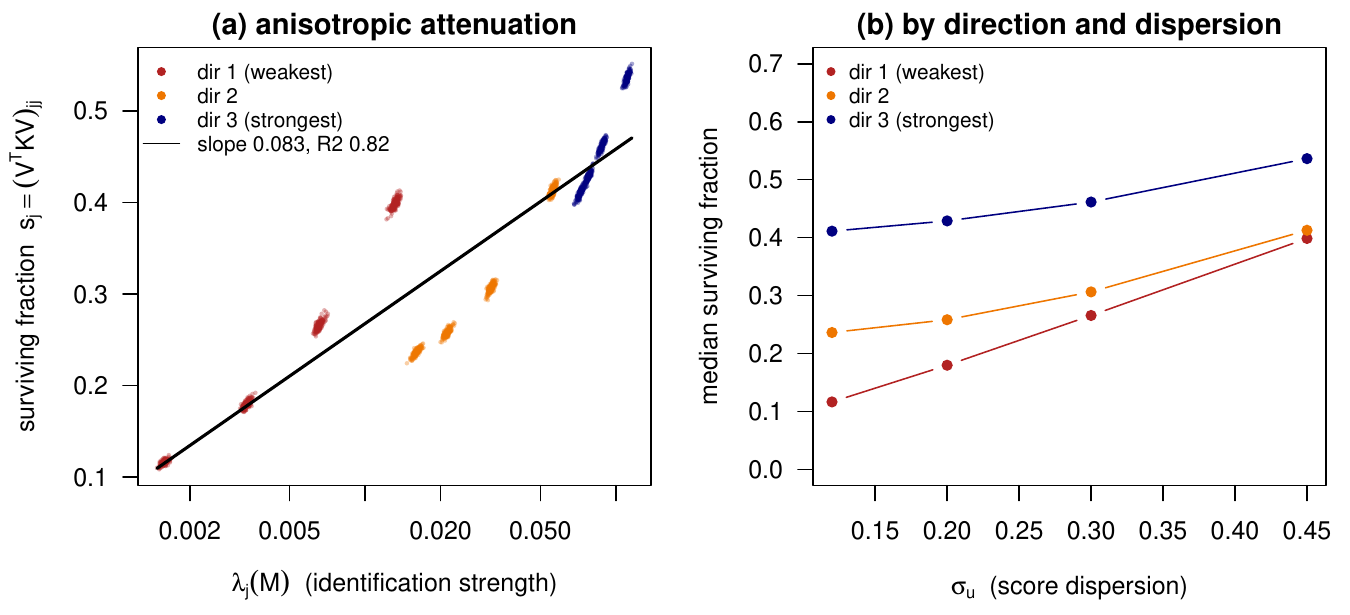}
\caption{E5. (a) Surviving fraction $s_j$ against the eigenvalue
$\lambda_j(\widehat\Msf)$, pooled over replicates and dispersions; the weakest
direction (red) is attenuated most. (b) Median $s_j$ by direction and
$\sigma_u$.}
\label{fig:e5}
\end{figure}

\section{Additional Empirical Results}

\subsection{Controlled validation on a real score spectrum (E8)}
We first check that the estimator recovers a known $\tau$ on a spectrum derived from real data. Using the real covariates and the real cross-fitted classifier score $p$, we draw membership $g\sim\mathrm{Multinomial}(p)$---so calibration holds by construction---and generate $Y=\mu(X)+g\T\tau+\varepsilon$ with a known $\tau=(0.8,0.5,-0.3)$ and no direct $W$-channel. The score's empirical spectrum is anisotropic (condition number $33$). Estimating from $p$ alone, the soft estimator recovers $\tau$ along every eigen-direction of $\widehat\Msf$: the standardized gaps to the truth are $|z|=0.86,0.89,1.44$ for the weakest to strongest direction, all within sampling error. Hard-labeling again attenuates anisotropically (retention coefficients $0.42,0.74,0.90$).

\section{A Land-Cover Audit}
\label{sec:covertype}

The Adult illustration uses a socio-economic score. We repeat the exercise on a different data geometry, and on one where the assumptions are less likely to hold, using UCI Forest CoverType \citep{blackard1999}, restricted to four cover types (\textit{Spruce/Fir}, \textit{Lodgepole Pine} as reference, \textit{Ponderosa Pine}, \textit{Krummholz}; $K=3$ contrasts, with Krummholz rare at about $4\%$). The latent class is the cover type, and the score $p$ is a cross-fitted random-forest probability vector trained on the wilderness-area and soil indicators $W$. The downstream controls are $X=(\text{elevation},\text{slope})$, deliberately kept out of $W$ so that the score retains residual variation once $X$ is partialed out.

\emph{The real score spectrum and its attenuation geometry} ($n=40{,}000$). Computed from the real covariates and the real classifier alone---no outcome and no calibration assumption is involved---the spectrum of $\widehat\Msf$ is anisotropic, $\lambda=(0.078,\,0.019,\,0.0099)$ (condition number $7.9$), and the retention coefficients under argmax hard-labeling are $s=(0.59,\,0.80,\,0.63)$ ordered from the strongest to the weakest eigen-direction. Here $s_j$ is non-monotone in $\lambda_j$ ($s_{\text{strong}}/s_{\text{weak}}=0.95$): a real score geometry illustrating the caveat after Corollary \MainCorAniso. The monotone attenuation-versus-eigenvalue pattern of the controlled designs (Section~\MainSecSim) is a feature of those designs, not a theorem, and the diagnostic must report the $s_j$ rather than presume their ordering.

\begin{sloppypar}
\emph{Known-truth recovery on the real spectrum.} We draw membership from the real score, so that calibration holds by construction, synthesize the outcome with known $\tau=(1.0,\,0.6,\,-0.5)$, and run $200$ replicates. The soft estimator recovers $\tau$ to $(0.998,\,0.603,\,-0.501)$, with RMSE $0.040$. The hard estimator is attenuated to $(0.590,\,0.423,\,-0.268)$ (RMSE $0.292$). The no-controls estimator, the analogue of regressing on the score without partialing out terrain, is confounded to the point of sign reversal ($(1.92,\,-1.73,\,1.21)$; RMSE $1.75$). The attenuation matrix computed from $(p,X)$ alone gives $\mathcal{A}\tau=(0.590,\,0.420,\,-0.268)$, within $0.003$ in every component of the realized hard-label estimate $(0.590,\,0.423,\,-0.268)$. As in E8, the operator closely predicts the hard-label estimate on a real score geometry.
\end{sloppypar}

\emph{A real-data audit} ($n=60{,}000$). Finally we run the same logic end to end: the outcome is a real terrain variable (horizontal distance to surface water), the classifier ($W=$ wilderness and soil indicators only; out-of-fold accuracy $0.678$) supplies $p$, the estimator never sees the labels, and the recorded cover types provide the full-information benchmark. Table~\ref{tab:covertype} reports the result together with a per-class conditional-calibration diagnostic (the $R^2$ of regressing $g_k-p_k$ on a quadratic basis in $X$; the richer battery in $(p,X)$, which reveals a substantially larger violation, is deployed in Section~\MainSecRecal). The Ponderosa contrast, whose calibration diagnostic is smallest, is recovered within sampling error ($0.332\pm0.018$ against a benchmark of $0.320$, $|z|=0.65$). The Spruce/Fir contrast is off by $|z|=14.6$, and the rare Krummholz class---with the worst calibration diagnostic, $R^2=0.091$---recovers only about a quarter of the benchmark effect. The soft estimator barely improves on the hard one overall (RMSE $0.344$ versus $0.361$).

\begin{table}[t]
\centering
\spacingset{1}
\caption{CoverType, real-data audit: cover-type contrasts (reference
\textit{Lodgepole Pine}) on distance-to-water. Benchmark uses the real labels.
Soft and hard use only the classifier score. cal-$R^2$: conditional-calibration
diagnostic for the class, quadratic basis in $X$ (larger is worse).}
\label{tab:covertype}
\begin{threeparttable}
\small
\begin{tabular}{lrrrrr}
\toprule
Contrast & benchmark & soft (SE) & hard & cal-$R^2$ & $|z|$ vs.\ benchmark\\
\midrule
Spruce/Fir vs.\ ref.  & $-0.501$ & $-0.659\ (0.011)$ & $-0.552$ & $0.021$ & $14.6$\\
Ponderosa vs.\ ref.   & $\phantom{-}0.320$ & $\phantom{-}0.332\ (0.018)$ & $\phantom{-}0.248$ & $0.012$ & $0.65$\\
Krummholz vs.\ ref.   & $-0.780$ & $-0.205\ (0.070)$ & $-0.161$ & $0.091$ & $8.2$\\
\bottomrule
\end{tabular}
\begin{tablenotes}\small
\item \textit{Notes}: $n=60{,}000$; classifier accuracy $0.678$; RMSE versus
the benchmark: soft $0.344$, hard $0.361$.
\end{tablenotes}
\end{threeparttable}
\end{table}

As in the Adult audit, now on a real outcome, validity requires both the calibration and exclusion assumptions, and an off-the-shelf classifier need not satisfy either. Where a labeled subsample is available, the regression of $g_k-p_k$ on $(p,X)$ can detect a calibration failure before the audit is reported rather than after, to the extent that the chosen basis has power against it. The spectrum of $\widehat\Msf$ carries no such information: it grades how well each direction is identified, not whether the assumptions hold. Section~\MainSecRecal examines the evidence on the source of the failure.


\section{The recalibration experiment on the two real-data audits}
\label{sec:recal-supp}
Table~\ref{tab:recal} collects the outcome for Adult and CoverType, alongside a separate probe of the exclusion component: the partial $R^2$ of the classifier features $W$ for the outcome, given the true class and $X$. The evidence points to different dominant problems in the two. On Adult there is barely any conditional miscalibration for the battery to find ($R^2\le 0.006$), so recalibration leaves the estimates where they were. The direct-channel check, by contrast, gives $\Delta R^2=0.135$. Occupation, work class, relationship, and capital income predict earnings on their own, so the failure of the weak directions is one that recalibration cannot address. Consistently, the argmax hard estimator matches the benchmark almost exactly (RMSE $0.023$, per-class $|z|\le 1.3$). This is consistent with hard-labeling discarding the within-class score variation through which the direct channel operates. Its argmax labels agree with the true class for $88.5\%$ of units. With an accurate classifier and a direct channel, the hard estimator is the closer of the two to the benchmark here, reversing the ordering that Theorem~\MainThmAtten implies when the exclusion component is clean.

For CoverType the rich basis reveals severe conditional miscalibration ($R^2$ up to $0.32$---far beyond what an $X$-only battery could see), and recalibration eliminates it out-of-sample ($R^2\le 0.003$). Yet the audit deteriorates. The largest eigenvalue falls from $0.122$ to $0.029$ under recalibration, and the violation moment is premultiplied by $\Msf^{-1}$ (Section~\MainSecRecal), so that compression enlarges the bias produced by the remaining direct channel---here the one carried by the soil and wilderness indicators ($\Delta R^2=0.081$).

\begin{table}[t]
\centering
\spacingset{1}
\caption{The recalibration experiment on the two real-data audits.
cal-$R^2$: largest per-class conditional-calibration $R^2$ on a rich
$(p,X)$ basis, before and after recalibration (out-of-sample). RMSE:
soft-estimator distance to the full-information benchmark on the audit
half. $\Delta R^2_{W}$: partial $R^2$ of the classifier features for the
outcome given the true class and $X$ (direct-channel check).}
\label{tab:recal}
\begin{threeparttable}
\begin{tabular}{lrrrrrr}
\toprule
& \multicolumn{2}{c}{max cal-$R^2$} & \multicolumn{2}{c}{RMSE soft} & &\\
\cmidrule(lr){2-3}\cmidrule(lr){4-5}
Audit & raw & recal & raw & recal & RMSE hard & $\Delta R^2_{W}$\\
\midrule
Adult (income)          & 0.006 & 0.006 & 0.397 & 0.463 & 0.023 & 0.135\\
CoverType (dist.\ water)& 0.317 & 0.003 & 0.395 & 1.322 & 0.393 & 0.081\\
\bottomrule
\end{tabular}
\begin{tablenotes}\small
\item \textit{Notes}: $50/50$ calibration/audit split ($n_{\mathrm{audit}}
= 15{,}081$ and $30{,}000$); recalibration map fitted on the labeled half
only. Eigenvalues of $\widehat\Msf$ compress under recalibration from
$(0.0055, 0.0248, 0.1216)$ to $(0.0031, 0.0110, 0.0287)$ on CoverType and
mildly on Adult.
\end{tablenotes}
\end{threeparttable}
\end{table}

The sequence of Section~\MainSecRecal starts from the rich-basis $(p,X)$ battery because an $X$-only check is blind to violations in the $p$-direction, as CoverType shows. Recalibration and an out-of-sample re-run follow. A closing gap to an available benchmark indicates calibration as the dominant failure, with the repaired audit then nearer the benchmark. If the audit barely moves despite a clean battery (Adult), or moves sharply after recalibration (CoverType), that is evidence of a remaining exclusion problem, which the direct-channel check can assess. Neither version of the soft audit should then be reported without it, and, when the classifier is accurate, the coarser argmax estimator may be closer to the benchmark.

Neither step tests the exclusion restriction directly, and neither is conclusive on its own. Together they help distinguish which of the two assumptions is more problematic, using only the labeled subsample the battery already requires. Section~\MainSecBisg turns to an audit on which both diagnostics are substantially more favorable.

\section{Spectral summaries and directional collapse}
\label{sec:e1e2-supp}
\subsection{A spectral summary tracks the variance (E1)}
We sweep $\sigma_u$ and regress, in logs, the largest diagonal of the sandwich covariance $\widehat V$ on each of three spectral summaries of $\widehat\Msf$. Table~\ref{tab:e1} shows that $\lambda_{\min}(\Msf)$ alone explains the variance to within rounding ($R^2=1.00$, slope $-1.19$), outperforming the condition number and $\tr(\Msf^{-1})$ as a single predictor. The slope differs from the naive $-2$ because the moment-noise matrix $\Sigma$ co-moves with $\Msf$ across designs. The slope is exactly $-2$ once $\Sigma$ is held fixed (experiment E3). Figure~\ref{fig:e1} plots the relationship.

\begin{table}[t]
\centering
\spacingset{1}
\caption{E1: spectral summaries of $\widehat\Msf$ versus the estimator's
variance, medians over replicates ($n=4000$, $K=3$). Log--log regressions of
$\max_k\widehat V_{kk}$ on each summary.}
\label{tab:e1}
\begin{threeparttable}
\begin{tabular}{lrrr}
\toprule
& \multicolumn{3}{c}{$\log\max_k\widehat V_{kk}$ regressed on $\log(\cdot)$}\\
\cmidrule(lr){2-4}
Spectral summary & slope & $R^2$ & \\
\midrule
$\lambda_{\min}(\Msf)$         & $-1.19$ & $1.00$ & (best single predictor)\\
condition number $\lambda_{\max}/\lambda_{\min}$ & $+1.71$ & $0.99$ & \\
$\tr(\Msf^{-1})$               & $+1.25$ & $1.00$ & \\
\bottomrule
\end{tabular}
\end{threeparttable}
\end{table}

\begin{figure}[t]
\centering
\includegraphics[width=0.56\textwidth]{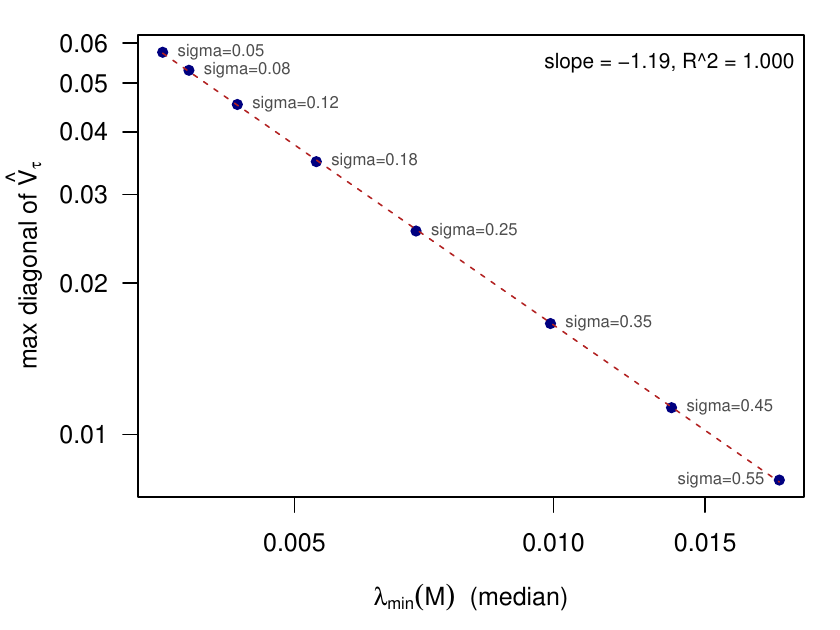}
\caption{E1. Largest diagonal of the sandwich covariance against
$\lambda_{\min}(\widehat\Msf)$ across the $\sigma_u$ sweep (log--log). The
single eigenvalue tracks the variance with $R^2=1.00$.}
\label{fig:e1}
\end{figure}

\subsection{Directional identification collapse (E2)}
We construct a design in which one group direction is served only by $X$ (so its score combination is nearly deterministic in $X$, driving its eigenvalue toward zero) while the others are served by $S$. At the default dispersion ($\sigma_u=0.30$, no tuning), Table~\ref{tab:e2} reports, in the eigenbasis of $\widehat\Msf$, the variance along the weakest versus strongest direction. The weak direction carries about seven times the variance of the strong one across the whole range of $\beta_x$, and, consistent with Theorem~\MainThmClt, the variance ratio $\Var_1/\Var_3$ tracks the eigenvalue ratio $\lambda_3/\lambda_1$ almost exactly ($7.0$ vs $7.0$, $6.1$ vs $6.1$), matching the approximately $\lambda_j^{-1}$ directional variance scaling in this design. Precision deteriorates along the weak direction while the others are unaffected. Figure~\ref{fig:e2} shows the per-direction variance against its eigenvalue.

\begin{table}[t]
\centering
\spacingset{1}
\caption{E2: directional collapse at the default $\sigma_u=0.30$ ($n=6000$),
no parameter tuning. Medians in the eigenbasis of $\widehat\Msf$. $\beta_x$ 
controls how strongly the weak direction is driven by $X$ alone.}
\label{tab:e2}
\begin{threeparttable}
\begin{tabular}{rrrrr}
\toprule
$\beta_x$ & $\lambda_1$ (weak) & $\lambda_3$ (strong) & $\Var_1/\Var_3$ &
$\lambda_3/\lambda_1$\\
\midrule
0.6 & $0.0060$ & $0.042$ & $7.0$ & $7.0$\\
1.0 & $0.0059$ & $0.040$ & $7.0$ & $6.9$\\
1.5 & $0.0058$ & $0.039$ & $6.8$ & $6.7$\\
3.0 & $0.0060$ & $0.036$ & $6.1$ & $6.1$\\
\bottomrule
\end{tabular}
\end{threeparttable}
\end{table}

\begin{figure}[t]
\centering
\includegraphics[width=0.56\textwidth]{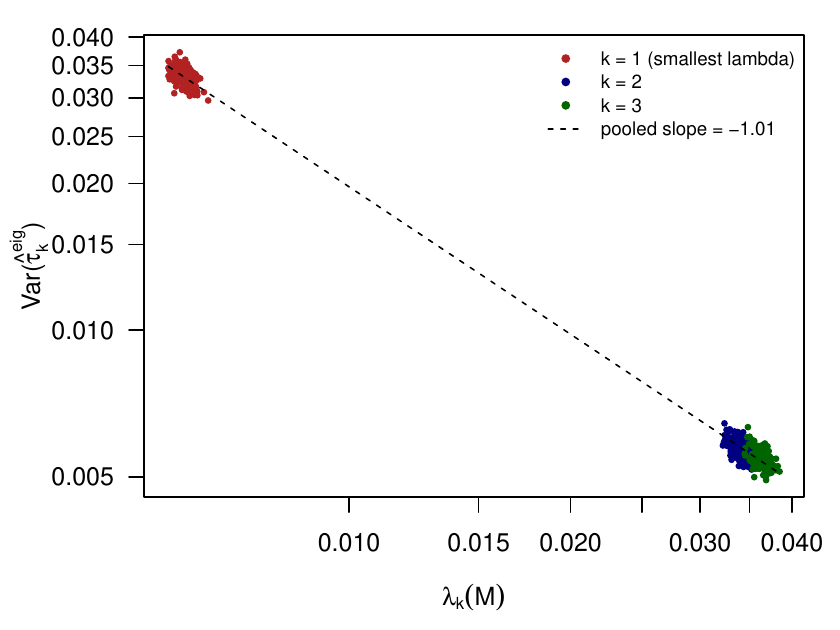}
\caption{E2. Per-eigen-direction variance of $\hat\tau$ against the eigenvalue
$\lambda_k(\widehat\Msf)$ (log--log). The weakest direction (red) has the
largest variance; the pooled slope is close to $-1$.}
\label{fig:e2}
\end{figure}

\section{Extended discussion of Theorems \MainThmSpectral and \MainThmEff}
\label{sec:extdisc}
Part~(c) of Theorem~\MainThmSpectral generalizes the scalar boundary $\Vs=0\Leftrightarrow p=r(X)$: a direction collapses precisely when the corresponding linear combination of scores carries no residual variation beyond $X$. Part~(b) shows that collapse is local to the affected directions---the effect remains point-identified along every direction with positive eigenvalue---and part~(d) upgrades the moment-level statement to genuine observational equivalence. The form of the construction in (d) matters. An alternative membership vector generated from $p$ by independent randomization satisfies $\E[Y\mid g',p,X]=\E[Y\mid p,X]$, which cannot depend on $g'$, so it supports no nonzero coefficient at all. A valid alternative model must couple membership to the outcome, and the proof does so by tilting the conditional class probabilities with a bounded, conditionally mean-zero transform of the outcome residual. Part~(d) also bounds the identified set of part~(b): a latent split of a fixed conditional outcome distribution can only support mean contrasts commensurate with its dispersion, so the identified set is bounded, not the full affine subspace. The conditions in (d) are sufficient for the bounded-tilt construction, not sharp. When $K=1$, $\Msf$ is a number and identification is all or nothing.

The efficiency bound rests on a Chamberlain calculation restricted to the instruments $h$ with $\E[h\mid X]=0$, the ones that leave the moment insensitive to the unknown $\mu$. Section~\ref{sec:proof-eff} shows that the optimum within that class is $\sigma^{-2}\tilde a$ and not $\sigma^{-2}a$. The restriction binds here because heteroskedasticity is intrinsic: since $u=R-a\T\tau=\tau\T(g-p)+\varepsilon$ and $\E[\varepsilon\mid G,p,X]=0$ kills the cross term,
\begin{equation}
  \sigma^{2}(p,X)=\tau\T\!\bigl(\operatorname{diag}(p)-pp\T\bigr)\tau
  +\Var(\varepsilon\mid p,X),
  \label{eq:het}
\end{equation}
so the latent-membership term generally varies with $p$ even when $\Var(\varepsilon\mid p,X)$ is constant, and the precision-weighted mean $\tilde p$ can then depart from $r$. Weighting can improve precision, but the efficient estimator requires two additional nuisance functions beyond $(m,r)$. Section~\ref{sec:e7} quantifies the gain from each step.

\section{Wald coverage for the uncoarsened estimator}
\label{sec:e6-supp}
\subsection{Multivariate inference is valid (E6)}
Table~\ref{tab:e6} reports componentwise and joint Wald coverage at nominal $95\%$ over $1000$ replicates, sweeping $n$. Componentwise coverage lies in $[0.938,0.960]$ and joint-ellipsoid coverage in $[0.946,0.958]$. A Shapiro--Wilk test on the standardized first component at $n=5000$ does not reject normality ($p=0.99$). Figure~\ref{fig:e6} shows the normal QQ-plot and the convergence of coverage to nominal, consistent with Theorem~\MainThmClt.

\begin{table}[t]
\centering
\spacingset{1}
\caption{E6: empirical coverage of nominal-$95\%$ Wald procedures
($1000$ replicates, $K=3$, $\sigma_u=0.30$).}
\label{tab:e6}
\begin{threeparttable}
\begin{tabular}{rrrrr}
\toprule
$n$ & $\tau_1$ & $\tau_2$ & $\tau_3$ & joint ellipsoid\\
\midrule
500   & 0.938 & 0.954 & 0.954 & 0.946\\
1000  & 0.946 & 0.945 & 0.954 & 0.947\\
2000  & 0.958 & 0.948 & 0.947 & 0.955\\
5000  & 0.954 & 0.960 & 0.951 & 0.958\\
\bottomrule
\end{tabular}
\end{threeparttable}
\end{table}

\begin{figure}[t]
\centering
\includegraphics[width=\textwidth]{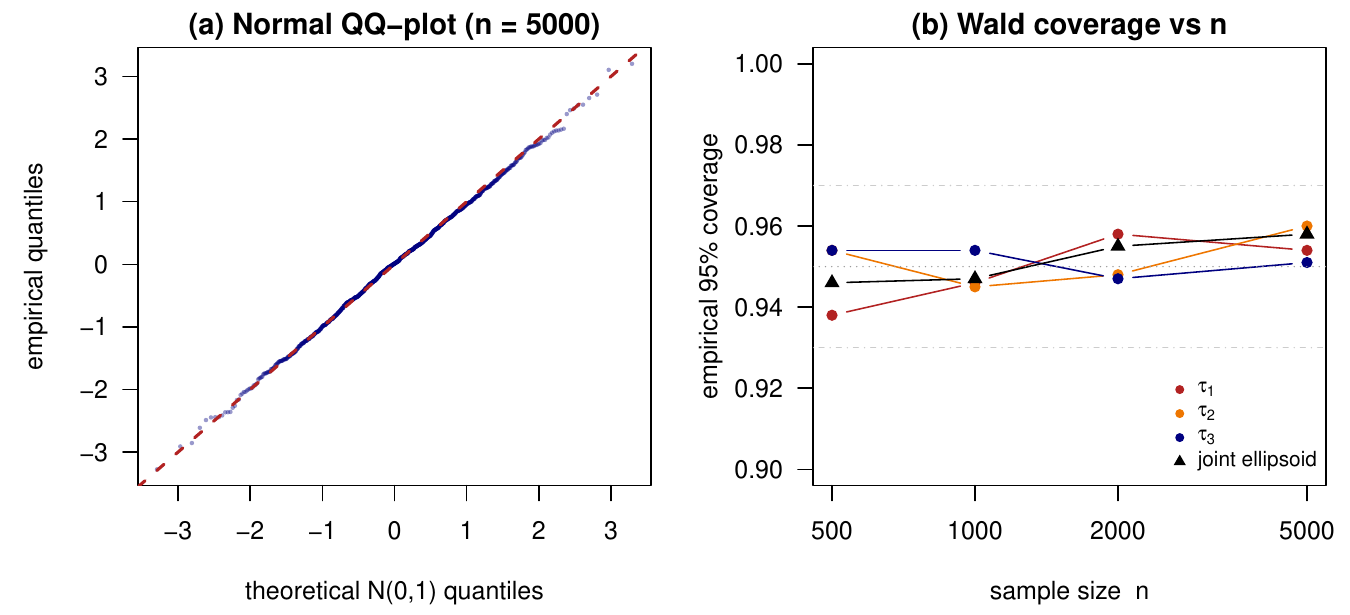}
\caption{E6. (a) Normal QQ-plot of the standardized $\hat\tau_1$ at $n=5000$.
(b) Componentwise and joint Wald coverage converging to the nominal $0.95$
(gray band $[0.93,0.97]$).}
\label{fig:e6}
\end{figure}

\section{The Adult disparity audit in full}
\label{sec:adult-supp}
The original UCI \textit{Amer-Indian-Eskimo} and \textit{Other} race categories are combined as \textit{Other}, giving the four classes used in the audit. This section reports in full the audit summarized in Table~\MainTabAdult. The spectrum of $\widehat\Msf$ is anisotropic, with eigenvalues $(0.0006,0.0069,0.0177)$ and condition number $31$, so one disparity direction carries substantially more residual variation than the other two. Hard-labeling attenuates anisotropically (Corollary~\MainCorAniso), an identity that requires no calibration assumption. The best-identified direction largely survives hard-labeling ($s_3=0.90$), while the weakest loses more than half its signal ($s_1=0.46$). Along the well-identified direction the soft estimator, built from the proxy alone, matches the full-information benchmark ($0.034\pm0.017$ versus $0.028$). Along the two weak directions the soft estimator departs sharply from the benchmark. With eigenvalues an order of magnitude smaller, these directions amplify violations of the identifying assumptions, so the disparity estimates cannot be interpreted under the identifying model without addressing those violations.

The diagnostics provide evidence about which assumption is more problematic. The calibration battery is clean both in $X$ (regressing $g_k-p_k$ on $X$ yields $R^2\le0.002$) and on the rich $(p,X)$ basis ($R^2\le0.006$), and the recalibration experiment of Section~\MainSecRecal shows that recalibration changes nothing. The evidence instead points to the exclusion restriction. The proxy's features (occupation, work class, capital income) carry direct income channels, a violation quantified by the direct-channel check and not addressed by recalibrating the score. These results suggest restricting interpretation of proxy-based disparities to well-identified directions and to settings in which both assumptions have been probed. This is a methodological stress test, not a substantive claim about disparities. The observed labels serve only as a benchmark. Proxy-based inference of a protected attribute is itself statistically and ethically fraught, and the weak directions here are dominated by assumption failure rather than by any credible estimate.

\subsection*{Voter-file category coding}
In the voter-turnout audit, registrants coded as Hispanic/Latino are assigned to the Hispanic category before the race categories are formed. Registrants not assigned to White, Black, Hispanic, or Asian are excluded from that audit.

\section{Calibration of the coverage formula}
\label{sec:cov-supp}
Proposition~\MainPropCoverage predicts coverage from the noncentrality index $\nu_v$. Grouping $600$ replicates at $n=4000$ by $|\nu_v|$ and comparing realized with predicted coverage within groups gives $0.945$ against $0.950$ for $|\nu_v|\le0.25$, $0.809$ against $0.812$ for $|\nu_v|\in(0.75,1.5]$, $0.542$ against $0.523$ for $|\nu_v|\in(1.5,2.5]$, and $0.003$ against $0.002$ for $|\nu_v|>4$. The formula tracks realized coverage to within $0.02$ across the entire range. Plugging $\hat\tau$ into $\hat\nu_v$ in place of $\tau$ opens the gaps further, especially where the distortion is largest. Equation~(\MainEqCovHat) therefore orders the candidates accurately in these simulations, while approximating the level less closely. Varying the score dispersion and the overlap between classifier features and controls does not disturb the ordering. The argmax interval covers between $1\%$ and $6\%$ in every configuration examined.

\section{Proof of Proposition \MainPropCoverage (coverage under coarsening)}
Fix a unit vector $v$ and write $b_v=v\T(\mathcal{A}_h-I)\tau$ and $\sigma^2_{h,v}=v\T\Omega_hv$ with $\Omega_h=D_h^{-1}\Sigma_hD_h^{-1}$, $\Sigma_h=\E[u_h^2a_ha_h\T]$ and $u_h=R-a_h\T\mathcal{A}_h\tau$. The estimator $\hat\tau_h$ solves the moment $a_h(R-a_h\T\vartheta)$, whose population solution is $\mathcal{A}_h\tau$ by Theorem~\MainThmAtten. The argument of Theorem~\MainThmClt applied to that moment, with $\vartheta$ in place of $\tau$, gives $\sqrt n\,(v\T\hat\tau_h-v\T\mathcal{A}_h\tau)\dto\N(0,\sigma^2_{h,v})$ and $\widehat\Omega_h\pto\Omega_h$.

\noindent The reported interval is $v\T\hat\tau_h\pm z_{1-\alpha/2}\sqrt{v\T\widehat\Omega_hv/n}$ and covers $v\T\tau$ when $|v\T\hat\tau_h-v\T\tau|\le z_{1-\alpha/2}\sqrt{v\T\widehat\Omega_hv/n}$. Writing $v\T\hat\tau_h-v\T\tau=(v\T\hat\tau_h-v\T\mathcal{A}_h\tau)+b_v$ and dividing by $\sqrt{v\T\widehat\Omega_hv/n}$, the studentized statistic converges to a normal variate with unit variance and mean $\nu_v=\sqrt n\,b_v/\sigma_{h,v}$. Along sequences with $b_v=b/\sqrt n$, the mean is $\nu_v=b/\sigma_{h,v}$ and the coverage probability converges to $\Phi(z_{1-\alpha/2}-\nu_v)-\Phi(-z_{1-\alpha/2}-\nu_v)$, which is the stated limit. If instead $b_v\neq0$ is held fixed, then $\nu_v\to\infty$ and the limit is $0$. The map $\nu\mapsto\Phi(z-\nu)-\Phi(-z-\nu)$ is strictly decreasing in $|\nu|$, so coverage falls as the coarsening bias grows relative to the reported standard error and, for a given bias, falls faster the smaller that standard error is.
\hfill$\square$

\section{Direction-specific rates and inference under collapse}
\label{sec:weakid}
\noindent The eigenvalue scaling of $\Omega$ has a sharp consequence for what happens as a direction collapses. Fix a unit vector $v$ and let $\lambda_{v,n}=v\T\Msf_n v$, allowing $\Msf_n$ to drift toward singularity along $v$.

\begin{proposition}[Direction-specific rate; collapse-robust inference; proved in Section~\ref{sec:weakid-proof}]
\label{prop:weakid}
Under Assumptions~1--4, conditional homoskedasticity $v\T\Sigma v=\sigma^2\lambda_{v,n}(1+o(1))$, and $n\lambda_{v,n}\to\infty$:
\begin{enumerate}[(i)]
\item $\sqrt{n\lambda_{v,n}}\,(v\T\hat\tau-v\T\tau)\dto\N(0,\sigma^2)$. The
convergence rate along $v$ is $\sqrt{n\lambda_{v,n}}$, i.e.\ the effective
sample size in direction $v$ is $n\lambda_{v,n}$.
\item The studentized statistic $(v\T\hat\tau-v\T\tau)/\sqrt{v\T\widehat V v}\dto\N(0,1)$,
so the Wald interval stays valid however slowly $\lambda_{v,n}\to0$.
\end{enumerate}
In the oracle Gaussian model (known nuisances), (ii) continues to hold even at the bounded-information rate $n\lambda_{v,n}\to\kappa\in(0,\infty)$. The estimate is then inconsistent, yet the Wald interval retains nominal coverage, widening at the $O(1)$ rate $(n\lambda_{v,n})^{-1/2}$.
\end{proposition}

Proposition~\ref{prop:weakid} gives the corresponding inference result and differs from weak-instrument asymptotics. Here inference need not be invalidated because the residualized score $a$ is its own exogenous instrument ($\E[au]=0$), while $\widehat\Msf^{-1}$ enters the estimate and its variance estimate identically and therefore cancels in the $t$-ratio (self-normalization).

The bounded-information case is more delicate. A finite local experiment need not yield first-order normality without further structure. We therefore state the result only for the oracle Gaussian design and check it numerically in E10. In this setting, the widening intervals along weak eigen-directions retain nominal coverage rather than becoming anti-conservative. Establishing uniformity over identification strength would call for the apparatus of \citet{han2019}, who obtain Wald tests with uniformly correct asymptotic size---including for one-dimensional subvectors---under a deficient-rank Jacobian. We do not attempt that here.